\documentclass[journal]{IEEEtran}
\usepackage[pdftex]{graphicx}  % for arXiv.org

\usepackage{amsmath}
\usepackage{amssymb}
\usepackage{amsthm}

\usepackage{braket}

\def\RR{\mathbb R}
\def\II{\mathbb I}

\newtheorem{Lmm}{Lemma}
\newtheorem{Thm}{Theorem}
\newtheorem{Dfn}{Definition}
\newtheorem{Crl}{Corollary}

\begin{document}

\title{Equivalence of three classical algorithms with quantum side information: Privacy amplification, error correction, and data compression}

\author{Toyohiro Tsurumaru%
\thanks{
The author is with Mitsubishi Electric Corporation, Information Technology R\&D Center,
5-1-1 Ofuna, Kamakura-shi, Kanagawa, 247-8501, Japan (e-mail: Tsurumaru.Toyohiro@da.MitsubishiElectric.co.jp).
This paper was presented in part at the 43rd Quantum Information Technology Symposium (QIT43), held online, Dec. 10-11, 2020, and at a poster session of QCrypt 2021 Conference, held online, Aug. 23-27, 2021.
}
}

\maketitle

\begin{abstract}
Privacy amplification (PA) is an indispensable component in classical and quantum cryptography.
Error correction (EC) and data compression (DC) algorithms are also indispensable in classical and quantum information theory.
We here study these three algorithms (PA, EC, and DC) in the presence of quantum side information, and show that they all become equivalent in the one-shot scenario.
As an application of this equivalence, we take previously known security bounds of PA, and translate them into coding theorems for EC and DC which have not been obtained previously.
Further, we apply these results to simplify and improve our previous result that the two prevalent approaches to the security proof of quantum key distribution (QKD) are equivalent.
We also propose a new method to simplify the security proof of QKD.
\end{abstract}


\section{Introduction}
Privacy amplification (PA) algorithms is an indispensable component in classical and quantum cryptography \cite{RennerPhD,VanAssche,tuyls2007security,MYCQZ16}.
The goal of PA is to generate a random bit string that is completely unknown to outside, from a bit string which may be partially leaked to outside.
On the other hand, error correction (EC) and data compression (DC) algorithms are also indispensable in classical and quantum information theory \cite{10.5555/1146355}.

We here study these three algorithms (PA, EC, and DC), in the presence of quantum side information, and show that they all become equivalent in the one-shot scenario.

The equivalence here means that the following two conditions are satisfied:
1) If one chooses any one of the three algorithms (PA, EC and DC) and its input, the remaining two algorithms and their inputs are also determined uniquely and automatically, and
2) The security or the performance indices of the three algorithms thus determined are all equal
(see Sec. \ref{sec:main_results} for the rigorous statements).
This means that these three algorithms are in fact a single algorithm viewed from different angles.

Also, by generalizing this result to randomized algorithms, we show the equivalence between a security evaluation method for PA (called the leftover hashing lemma, LHL), and the coding theorems for EC and DC.
From a practical viewpoint, this means, e.g., that if one wishes to improve PA algorithms used for quantum cryptography, it suffices to improve DC or EC algorithms instead, and vice versa.

The equivalence is made possible by modifying these three algorithms in the following three points:

First, we consider the generalized cases where there is quantum side information:
In PA, information leaked to the eavesdropper is not necessarily classical but quantum \cite{RennerPhD,5961850}.
In EC and in DC, the decoder can use an auxiliary quantum state, in addition to the classical codeword or the usual classical compressed data \cite{doi:10.1098/rspa.2010.0445,8970489}.

Second, in order to evaluate the security of PA, we use a relatively new security index put forward by K\"{o}nig et al. \cite{5208530} based on the purified distance \cite{TomamichelPhD}, while most literature use the conventional criterion called the $\varepsilon$-security (see e.g. Ref. \cite{RennerPhD}), based on the trace distance.
We stress that we do not lose the security essentially by using this new criterion;
see Section \ref{sec:justification} for details.

Third, we restrict hash functions used for PA, and codes used for PA and EC to be linear.

We also demonstrate the usefulness of this equivalence with two applications:

First, we take previously known security bounds (i.e., LHLs) of PA, and convert them into new coding theorems for EC and DC with quantum side information.
Specifically, we consider three types of hash functions $F$ which are widely used for PA, namely, the universal$_2$ \cite{CARTER1979143}, the {\it almost} universal$_2$ \cite{WEGMAN1981265}, and the {\it almost dual} universal$_2$ functions \cite{FS08,6492260}.
Then we convert their LHLs into the coding theorems of EC and DC using the dual function of $F$.
To the best of our knowledge, these coding theorems are new results that have not been obtained previously.

Second, we apply these results to the security proof of quantum key distribution (QKD) \cite{Nielsen-Chuang}.
In the field of QKD, there are two major approaches to the security proof, called the leftover hashing lemma (LHL)-based approach \cite{RennerPhD}, and the phase error correction (PEC)-based approach \cite{Mayers98,Lo2050,SP00,Koashi,H07}.
Previously, we have shown that these two approaches are in fact equivalent mathematically \cite{8970489}.
In this paper, we simplify and improve this proof by exploiting the equivalence of the three algorithms (PA, EC, and DC).
The proof here is improved in that it is valid for a larger class of hash functions; that is, the equivalence holds for the case where the random function $F$ for PA is {\it almost} universal$_2$ and {\it almost dual} universal$_2$, while previously we treated only the case of universal$_2$ \cite{8970489}.

Further, utilizing the knowledge gained in this new proof, we propose a method to simplify the PEC-based proof.
That is, we propose to evaluate the randomness of Alice's phase degrees of freedom by the smooth max-entropy, rather than by the phase error rate.
This method has an additional merit that every step of the proof becomes equivalent to that of  the corresponding LHL-based proof.
As a result, one is guaranteed to reach exactly the same security bound as in the LHL-based approach, without any extra factor.

The relations between results of the present manuscript and the existing literature are as follows:

In Refs. \cite{PhysRevA.78.032335, 6157080}, Renes and coauthors studied PA and DC in the independent and identically distributed (i.i.d.) setting, not in the one-shot scenario, and showed that their asymptotic rates are equal.
% revised 2021/10/05 from here
Renes refined these results in more recent papers \cite{8047296, 8434309}, and showed that PA and DC are equivalent in the one-shot scenario using a fixed (not randomized) code (Ref. \cite{8047296}, Corollary 11).
This corresponds to a limited case of Theorem 1 of the present manuscript.
In comparison, our contributions here are the extension of this equivalence that includes EC (Theorem 1), and all the results after that, such as (i) we generalized Theorem 1 further to random codes, and showed the equivalence of LHL for PA and the coding theorems for DC and EC (Theorems 2 and 3), (ii) obtained explicit forms of equivalent pairs of an LHL and a coding theorem for practical cases including the universal$_2$ hash function (Section V), and (iii) showed how these results can simplify and improve the security proof of QKD (Section VI).
%Subsequently in Ref. \cite{doi:10.1098/rspa.2010.0445}, he also studied the relation (duality) between PA and DC in the one-shot scenario, but his bounds were not tight enough to prove their equivalence (see Section \ref{sec:subsec_equivalence_3_algorithms} for details).
% revised 2021/10/05 to here

We also note that our main result here can be considered as a refinement of the results of our previous paper \cite{8970489}.
Previously, we have shown the leftover hashing lemmas (LHLs) can be derived from a coding theorem of EC with quantum side information \cite{8970489}.
% revised 2021/10/05 from here
On the other hand, in the present manuscript, we demonstrate that PA and EC are not only directly connected, but are in fact equivalent, if we define the security of PA using the purified distance, instead of the trace distance which has been widely used.
% revised 2021/10/05 to here
In addition, we also prove that DC with quantum side information is also equivalent to these two algorithms.

\section{Notation}
\label{sec:notation}

All the rules below apply to alphabets besides $A$, $U$ and $V$.

\subsection{Random variables and the Hilbert spaces}
We denote random variables by a capital letter, such as $A$.
The same capital letter will also be used to denote the Hilbert space where the random variable is stored, unless otherwise specified.
For example, if a random variable, $A$, is already defined and if we speak of Hilbert space $A$, it means that random variable $A$ is stored in Hilbert space $A$ (for examples of this notation, see Refs.  \cite{RennerPhD,TomamichelPhD}).

\subsection{Use of tilde}
\label{sec:use_of_tilde}
If a random variable $A$ takes value $a$, and if it is stored in the $z$ basis, we write the situation as $\ket{a}_A$ without tilde; and if stored in the $x$ basis, we write it as $\ket{\tilde{a}}_A$ with tilde.

To put it more precisely: We denote the Pauli matrices in the $z$ and the $x$ bases by $\sigma_Z$ and $\sigma_X$ respectively\footnote{%
It is straightforward to generalize all our results below to the mutually unbiased bases (MUB) in prime power dimensions.
However for the sake of simplicity, in this paper we limit ourselves with the case of qubits.}.
We denote eigenstates of $\sigma_Z$ and $\sigma_X$ in a qubit space by $\ket{z}$ and $\ket{\tilde{x}}$ respectively; i.e., $\sigma_Z\ket{z}=(-1)^z\ket{z}$ and $\sigma_X\ket{\widetilde{x}}=(-1)^x\ket{\widetilde{x}}$ with $z,x\in\{0,1\}$.
We also extend this notation to multi-qubit space with length $l$, and write $\ket{z}=\ket{z_1}\otimes\cdots\otimes\ket{z_l}$ and $\ket{\widetilde{x}}=\ket{\widetilde{x_1}}\otimes\cdots\otimes\ket{\widetilde{x_l}}$ for $z=(z_1,\dots,z_l), x=(x_1,\cdots,x_l)\in\{0,1\}^l$.

%The measurement of eigenvalues of $Z$ ($X$) is called the the $z$-basis ($x$-basis) measurement.

\subsection{Classical states}
Given any quantum state $\rho_U$ in a space $U$, we denote by $Z^U$ the random variable that results from the $z$-basis measurement in $U$, and by $\rho_{Z^U}$ the resulting state; see, e.g., Ref. \cite{doi:10.1098/rspa.2010.0445}.

In this notation, when given a state $\rho_{UV}$, the result of $z$ basis measurement on space $U$ takes the form
\begin{equation}
\rho_{Z^UV}=\sum_z(\ket{z}\bra{z}_U\otimes\II_V)\rho_{UV}(\ket{z}\bra{z}_U\otimes\II_V),
%&=:&\sum_z\ket{z}\bra{z}_A\otimes\rho_{B}^z
\end{equation}
with $\II$ being the identity operator.
We say that $\rho_{UV}$ is {\it classical in} $Z^U$ if $\rho_{UV}=\rho_{Z^UV}$, i.e., if $\rho_{UV}$ is invariant under the $Z^U$-measurement (or informally, if $U$ is already measured in the $z$ basis).

We also use the same notation for the $x$ basis; e.g., $\rho_{UV}$ is classical in $X^U$ if $\rho_{UV}=\rho_{X^UV}$.

\subsection{Distance measures and entropies}

A state $\rho$ is called sub-normalized if ${\rm Tr}(\rho)\le1$.
The $L_1$ norm of a matrix $A$ is $\|A\|_1:={\rm Tr}(\sqrt{AA^\dagger})$.
For two sub-normalized states $\rho,\sigma$, the generalized fidelity is $F(\rho,\sigma):=\|\sqrt{\rho}\sqrt{\sigma}\|_1+\sqrt{(1-{\rm Tr}(\rho))(1-{\rm Tr}(\sigma))}$, and the purified distance is $P(\rho,\sigma):=\sqrt{1-F(\rho,\sigma)^2}$.
We say that $\rho,\sigma$ are $\varepsilon$-close and write $\rho\approx_\varepsilon\sigma$, if $P(\rho,\sigma)\le\varepsilon$ (see e.g. Ref. \cite{TomamichelPhD}).

The conditional min- and max-entropies of a sub-normalized state $\rho_{UV}$ are
\begin{align}
&H_{\rm min}(U|V)_\rho\nonumber\\
&\quad:=-\log\min_{\sigma\ge0}\{{\rm Tr} (\sigma)\,:\,\rho_{UV}\le\II_U\otimes\sigma_V\},\\
&H_{\rm max}(U|V)_{\rho}\nonumber\\
&\quad:=\max_{\sigma\ge0,{\rm Tr}(\sigma)=1}\log_2 (|{\cal U}|F(\rho_{UV},|{\cal U}|^{-1}\II_U\otimes\sigma_V)^2),
\end{align}
where ${\cal U}$ denotes the domain (alphabets) of random variable $U$, and $|{\cal U}|$ its cardinality.
Their smoothed versions are
\begin{align}
H_{\rm min}(U|V)^\varepsilon_\rho&:=\max_{\bar{\rho}}H_{\rm min}(U|V)_{\bar{\rho}},\\
H_{\rm max}(U|V)^\varepsilon_{\rho}&:=\min_{\bar{\rho}}H_{\rm max}(U|V)_{\bar{\rho}},
\end{align}
where the maximum and the minimum are evaluated for sub-normalized states $\bar{\rho}_{UV}\approx_\varepsilon\rho_{UV}$
(see e.g. Ref. \cite{TomamichelPhD}).

\section{Three algorithms to be considered}
\label{sec:def_algorithms}

In this section, we specify three algorithms to be considered in this paper: Privacy amplification (PA), error correction (EC), and data compression (DC).
In the next section, these three algorithms will be shown equivalent to each other.

\subsection{Privacy amplification (PA) against quantum side information}
\label{sec:PA}
Privacy amplification (PA) is an algorithm for extracting a secret random bits, from a bit string which may be partially leaked outside (Fig. \ref{fig:privacy amplification}) \cite{RennerPhD,5961850,VanAssche,tuyls2007security}.

\subsubsection{Description of the algorithm}
By definition, A PA algorithm starts with a situation where the legitimate user, Alice, has a classical information $z\in\{0,1\}^n$, and the eavesdropper, Eve, has a quantum state (quantum side information) $\rho^z$ correlated with $z$.
That is, there is initially a classical-quantum (cq) state,
\begin{equation}
\rho_{Z^AE}=\sum_z\ket{z}\bra{z}_A\otimes\rho_E^z,
\label{eq:rho_AE_defined}
\end{equation}
where Hilbert spaces $A$ and $E$ describe Alice's and Eve's degrees of freedom, respectively.
Then Alice applies a linear function $f:\{0,1\}^n\to\{0,1\}^m$ to $z$ and generates a random bits $k=f(z)$.
The function $f$ used in this context of PA is often called a {\it hash function}.
As a result, random bits $K$  and Eve end up in a state
\begin{equation}
\rho_{KE}^f:=\rho_{f(Z^A)E}=\sum_{k}\ket{k}\bra{k}_K\otimes\sum_{z\in f^{-1}(k)}\rho^z_E,
\label{eq:rho_f_KE_def}
\end{equation}
where $k$ is the output of function $f$, and $K$ is the Hilbert space for storing $k$.

\begin{figure}[t]
 \includegraphics[bb=0 0 950 450, width=\linewidth]{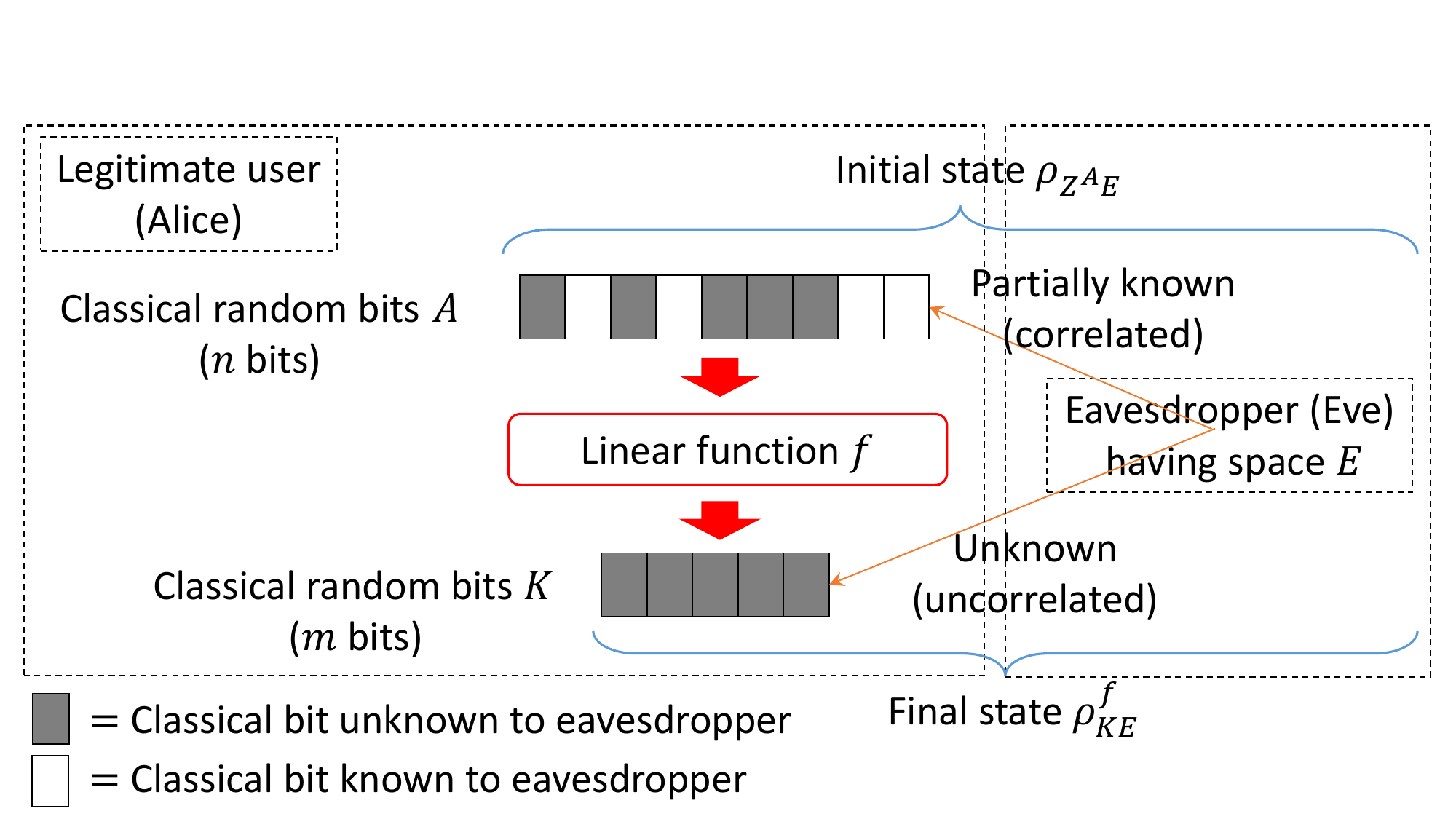}
 \caption{
Privacy amplification (PA) against quantum side information.
Initially, Alice's bit string $Z^A$ is partially known to the Eve; i.e., Eve has a quantum side information in space $E$, which is correlated with $Z^A$.
By applying a linear function $f$, Alice attempts to convert $Z^A$ into a bit string $K$, which is unknown to Eve.
}
 \label{fig:privacy amplification}
\end{figure}

\subsubsection{Security criteria of PA}
We say that random bits $K$ is ideally secure, if it is uniformly distributed and completely unknown to Eve, i.e., if $\rho_{KE}^f=2^{-m}\II_K\otimes\sigma_E$ for a normalized state $\sigma$.
However, such ideal state is unrealistic in practice.
Thus it is customary to define the security by the distance between the actual state $\rho^f_{KE}$ and the ideal state $2^{-m}\II_K\otimes\sigma_E$.

In this paper, we particularly follow K\"{o}nig et al. \cite{5208530} and measure the security by an index
\begin{eqnarray}
Q^{{\rm PA},f}(\rho_{Z^AE})&:=&{\rm Tr}(\rho_{Z^AE})-2^{H_{\rm max}(f(Z^A)|E)_{\rho}-m}\nonumber\\
&=&{\rm Tr}(\rho^f_{KE})-2^{H_{\rm max}(K|E)_{\rho^f}-m}.
\label{eq:QPA_H_max}
\end{eqnarray}
The idea here is to measure the distance between the actual and ideal states by the square of the purified distance.
Note that if $\rho_{Z^AE}$ is normalized, $Q^{{\rm PA},f}(\rho_{Z^AE})$ indeed equals the squared purified distance between the two states.
For the general case where $\rho_{Z^AE}$ is not necessarily normalized, $Q^{{\rm PA},f}(\rho_{Z^AE})$ is defined to scale proportionally to ${\rm Tr}(\rho_{Z^AE})$.
%That is, $Q^{{\rm PA},f}(\rho_{Z^AE})=\min_{\sigma\ge0,{\rm Tr}(\sigma)=1}P(\rho_{KE}^f, 2^{-m}\II_K\otimes\sigma_E)^2$, if $\rho_{KE}^f$ is normalized.
%\begin{equation}
%Q^{{\rm PA},f}(\rho_{Z^AE})=\min_{\sigma\ge0,{\rm Tr}(\sigma)=1}P(\rho_{KE}^f, 2^{-m}\II_K\otimes\sigma_E)^2.
%\label{eq:new_security_criteria}
%\end{equation}

\subsubsection{Justification for using the new index $Q^{{\rm PA},f}$}
\label{sec:justification}

\paragraph{Conventional criterion}
\label{sec:relation_with_the_conventional}
On the other hand, in fact, most existing literature on PA do not use the security index $Q^{{\rm PA},f}$ of (\ref{eq:QPA_H_max}).
They rather use an alternative index
\begin{equation}
d_1(\rho_{KE}^f):=\left\|\rho_{KE}^f-2^{-m}\II_K\otimes\rho_E\right\|_1
\label{eq:conventional_security_criteria}
\end{equation}
based on the trace distance, %That is, they use the trace distance (see e.g. Ref. \cite{Nielsen-Chuang}) instead of the purified distance, and also fix the ideal state $2^{-m}\II_K\otimes\sigma_E$ to be $2^{-m}\II_K\otimes\rho_E$.
and say that random number $K$ is $\varepsilon$-secure if $\frac12d_1(\rho_{KE}^f)\le\varepsilon$ (e.g. Ref. \cite{RennerPhD}).
This criterion, $\varepsilon$-security, is prevalent because it is explicitly shown to satisfy a desirable property called the universal composability \cite{10.1007/978-3-540-30576-7_21}.

%On the other hand, our use of a rather new criterion (\ref{eq:new_security_criteria}) is motivated by the fact that it allows to prove the equivalence of PA and other two algorithms, EC and DC.
Even so, our use of a rather new criterion, $Q^{{\rm PA},f}$, can be justified by the following two observations.

\paragraph{$\varepsilon$-security using the new index $Q^{{\rm PA},f}$}
%As already mentioned in K\"{o}nig et al. \cite{5208530}, 
The $\varepsilon$-security can also be guaranteed by using $Q^{{\rm PA},f}(\rho_{Z^AE})$.
This is because $Q^{{\rm PA},f}(\rho_{Z^AE})$ bounds $d_1(\rho^f_{KE})$ as
\begin{equation}
 d_1(\rho_{KE}^f) \le 4\sqrt{{\rm Tr}(\rho)}\sqrt{Q^{{\rm PA},f}(\rho_{Z^AE})}
\label{eq:d1_Q_bound}
\end{equation}
(see Appendix \ref{sec:security_Eriteria} for the proof).

\paragraph{Tightness of security bounds}
\label{par:tightness}

Bounds on $Q^{{\rm PA},f}$ thus obtained are nearly as tight as previously obtained bounds on $d_1(\rho^f_{KE})$ in many practical situations.
%That is, one might argue that the bounds on $d_1(\rho^f_{KE})$ obtained via (\ref{eq:d1_Q_bound}) may be looser than those obtained by analyzing $d_1(\rho^f_{KE})$ directly.
%However, that is not the case for many practical situations.
%; the bounds thus obtained are nearly as tight as the conventional ones.

For example, if we let function $f$ be a random function called the universal$_2$ function \cite{CARTER1979143}, and denote it by capital letter $F$ (see Section \ref{sec:randomizing_fg} for details of this notation), then we have
\begin{equation}
{\rm E}_F\, Q^{{\rm PA},F}(\rho_{Z^AE})\le2^{m-H_{\rm min}(Z^A|E)_\rho},
\label{eq:bound_on_EFQ}
\end{equation}
where the expected value ${\rm E}_F$ is taken on the ensemble of random function $F$ (Lemma 12 of Ref. \cite{8970489}, or Lemma \ref{lmm:LHL1} of the present paper).
%In (\ref{eq:bound_on_EFQ}), function $f$ is considered as a random variable and denoted by $F$; see Section \ref{sec:randomizing_fg} for details of this notation.
If we further apply (\ref{eq:d1_Q_bound}) and Jensen's inequality to (\ref{eq:bound_on_EFQ}), we obtain
\begin{eqnarray}
{\rm E}_F\, d_1(\rho_{KE}^F)
&\le& 4\sqrt{{\rm Tr}(\rho)}\sqrt{2^{m-H_{\rm min}(Z^A|E)_\rho}},
\label{eq:LHL_universal_2}
\end{eqnarray}
which differs only by a factor of 4 from the well-known bound called the leftover hashing lemma (LHL, or Theorem 5.5.1 of Ref. \cite{RennerPhD}.
Later, we will derive it again as Eq. (\ref{eq:original_LHL})).
Note that this factor 4 is harmless in practice, since it can be compensated for by reducing the length of random bits $m$ only by 4 bits.

In Section \ref{sec:application1_LHL_coding_theorem}, we will also show that similarly tight bounds can be obtained for almost universal$_2$ \cite{WEGMAN1981265} and almost {\it dual} universal$_2$ hash functions \cite{FS08,6492260}.

\subsubsection{LHLs of the conventional type and of the new type}
\label{sec:LHL_new_type}
In the past literature, an LHL always meant a bound on an average of the conventional security index, ${\rm E}_F\, d_1(\rho_{KE}^F)$.
In this paper, we extend this terminology and use the word LHL to also mean a bound on an average of the new security index, ${\rm E}_F\, Q^{{\rm PA},F}(\rho_{Z^AE})$ (e.g. (\ref{eq:bound_on_EFQ})).
When we need to distinguish between these two types, we call the former an LHL of the conventional type, and the latter the new type.

% modified 2021/10/07 from here
We do not claim that the new security index based on $Q^{{\rm PA},F}$ is either best or proper.
The main motivation for using this index here is to clarify the equivalence of PA and other two algorithms, EC and DC, defined below.
% modified 2021/10/07 to here

\subsection{Error correction (EC) with quantum side information}
\label{sec:EC_side_info}

Next we introduce a generalized form of classical error correction (EC), which we call {\it EC with quantum side information} (Fig. \ref{fig:error_correction}).
This is identical to what we called the {\it generalized error correction} in Ref. \cite{8970489}, but in this paper we will use the name above to clarify the relation with the data compression algorithm to be discussed in the next subsection.
From now on, whenever we say EC, we mean this generalized form.

In a conventional classical EC algorithm \cite{10.5555/1146355}, the sender chooses a message $t\in\{0,1\}^m$ and generates the corresponding codeword $c(t)\in\{0,1\}^n$.
The string $c(t)$ is then sent through a noisy channel, and output as a string $x$, which is $c(t)$ with bit flips applied probabilistically.
The receiver then decodes $x$ to recover  $c(t)$.

% revised 2021/10/07 from here
On the other hand, in our EC with quantum side information, codeword $c(t)$ is sent through a noisy quantum channel and output as a cq state $\rho^{c(t)}=\sum_x\ket{\tilde{x}}\bra{\tilde{x}}\otimes\tilde{\rho}^{c(t),x}$; see Fig. \ref{fig:error_correction}.
Or equivalently as a classical string $x$ {\it plus} an auxiliary quantum state (quantum side information) $\tilde{\rho}^{c(t),x}$, where the tilde indicates that it is an expansion of  $\rho^{c(t)}$ with respect to $\ket{\tilde{x}}$.
Hence, it is possible that the decoding succeeds with a higher probability, due to information (i.e. hint) obtained by measuring $\tilde{\rho}^{c(t),x}$.
Details are as follows.
% revised 2021/10/07 to here

\subsubsection{Description of the algorithm}
For the sake of similicity, we assume that the channel is symmetric under bit flips (i.e., binary symmetric channel).
% Recall that this symmetry can always be realized by twirling.
We also assume that the error correcting code $C\subset\{0,1\}^n$ is linear, i.e., a linear $[n,m]$ code.
We denote its linear syndrome function by $g:\{0,1\}^n\to\{0,1\}^{n-m}$.

\paragraph{Encoding and sending}
The sender chooses an message $m\in\{0,1\}^m$, and the classical codeword $c(t)\in C$ that corresponds to it.

\paragraph{Quantum channel}
% revised 2021/10/07 from here
The codeword $c(t)$ is input to the quantum channel and then output as a cq state $\rho^{c(t)}=\sum_x\ket{\tilde{x}}\bra{\tilde{x}}\otimes\tilde{\rho}^{c(t),x}$.
Or equivalently it is output as a classical string $x$ plus an auxiliary quantum state  $\tilde{\rho}^{c(t),x}$.
% revised 2021/10/07 to here

\paragraph{Decoding}
\label{sec:decoding_ec}
The receiver performs the following decoding algorithm using $x$ and  $\tilde{\rho}^{c(t),x}$.
\begin{enumerate}
\item Calculate the syndrome $s=g(x)\in\{0,1\}^{n-m}$.
\item Measure $\tilde{\rho}^{c(t),x}$ using positive operator valued measures (POVMs) which depend on $s$; that is,
$M^s=\{M^{s,e}\,|\,e\in\{0,1\}^n\}$ satisfying $\sum_eM^{s,e}=\II$.
The result $e$ is the estimated error pattern.
\item Output $y=x+e$ as the estimated codeword.
\end{enumerate}

\begin{figure}[t]
 \includegraphics[bb=0 0 950 520, width=\linewidth]{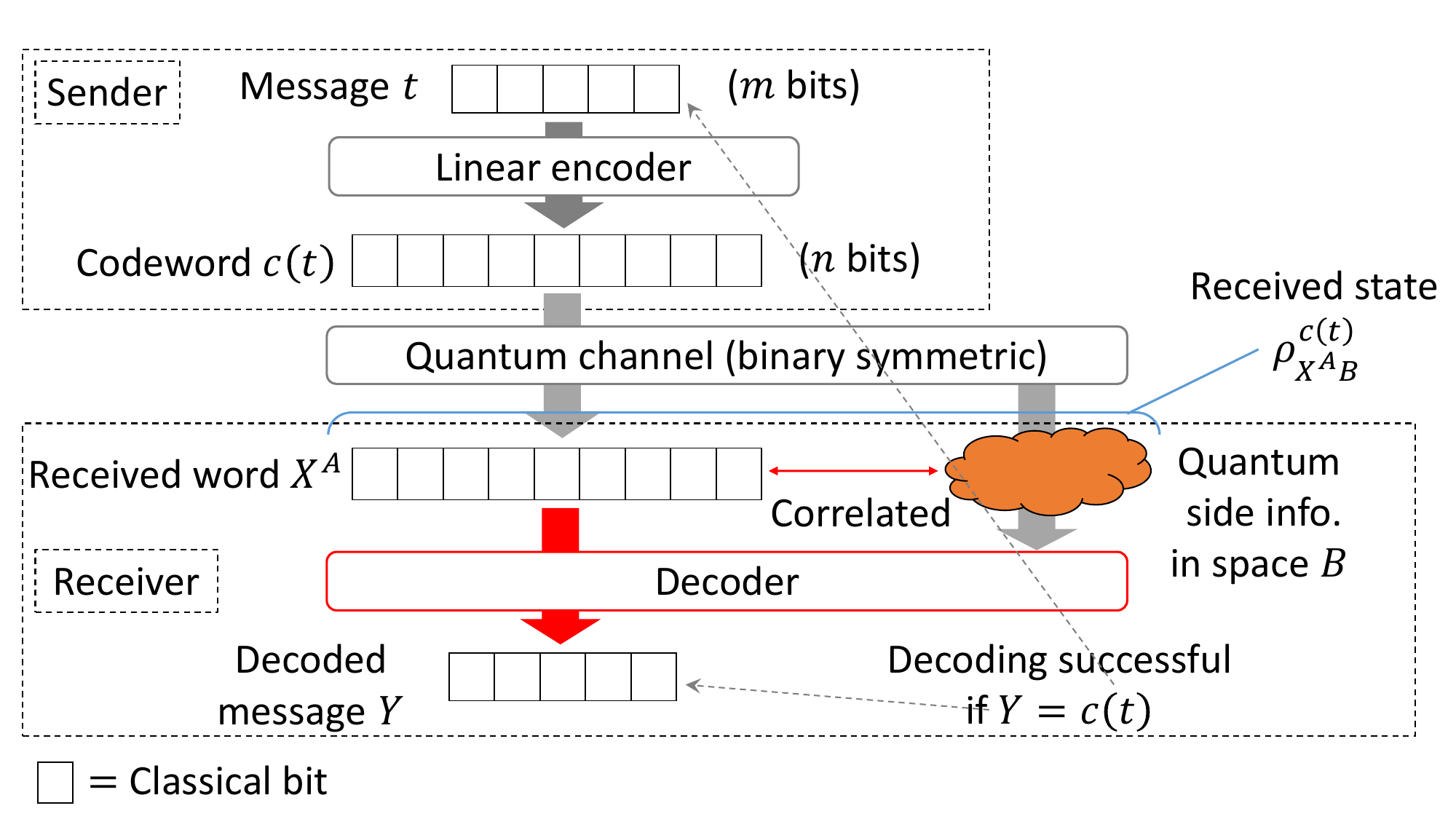}
 \caption{
Error correction (EC) with quantum side information.
The sender uses a linear classical error correcting code $C$, and we denote its syndrome function by $g$.
Note that the situation becomes the same as in the conventional classical error correction, if the quantum side information (orange cloud) is absent.
}
 \label{fig:error_correction}
\end{figure}

\subsubsection{Performance index}
We say that the decoding is successful when its output $y$ equals the correct codeword $c(t)$; i.e.,
$\Pr[{\rm EC\ succeeds}]:=\Pr[Y=c(T)]$.
We then define the performance index of our EC to be its failure probability (block error rate) when using the best decoder,
\begin{equation}
Q^{{\rm EC},g}:=\min_{\{M^s\}}\Pr[Y\ne c(T)].
\end{equation}

As we here limit ourselves with a BSC, we have
\begin{equation}
\tilde{\rho}^{c(t)+\Delta,x+\Delta}=\tilde{\rho}^{c(t),x}
\label{eq:BSC_condition}
\end{equation}
for any $\Delta \in\{0,1\}^n$.
Hence, as we also limit ourselves with a linear code $C$, it suffices to consider $Q^{{\rm EC},g}$ for conditioned on message $t=0$ and the codeword $c(0)=0$; that is,
\begin{eqnarray}
Q^{{\rm EC},g}(\rho_{X^AB}^0)&=&\min_{\{M^s\}}\Pr[Y\ne 0\,|\,T=0]\nonumber\\
&=&\min_{\{M^s\}}\Pr[Y\ne 0\,|\,\rho^0_{X^AB}],
%&=&\min_{\{M^s\}}\left(1-\sum_{x}{\rm tr}(\tilde{\rho}^{x} M^{g(x),x})\right).
\label{eq:QEC_def_with0}
\end{eqnarray}
where
\begin{equation}
\rho_{X^AB}^0=\sum_x\ket{\tilde{x}}\bra{\tilde{x}}_A\otimes\tilde{\rho}_B^{0,x}
\label{eq:rho_AB_EC_with0}
\end{equation}
denotes the output of the quantum channel on input $c(0)=0$.
As remarked in Section \ref{sec:use_of_tilde}, $\ket{\tilde{x}}$ appearing in (\ref{eq:rho_AB_EC_with0}) means that the channel output $x$ is encoded in the $x$ basis.
We stress that we lose no generality by using this particular choice of the basis; see Section \ref{sec:note on the basis choices}.

In what follows, for ease of notation, we will often omit superscript 0 of $\rho_{X^AB}^0$ and $\tilde{\rho}_B^{0,x}$, and write
\begin{equation}
\rho_{X^AB}=\rho_{X^AB}^0,\quad \tilde{\rho}_B^{x}=\tilde{\rho}_B^{0,x}.
\end{equation}
In this notation, Eq. (\ref{eq:rho_AB_EC_with0}) is rewritten as
\begin{equation}
\rho_{X^AB}=\sum_x\ket{\tilde{x}}\bra{\tilde{x}}_A\otimes\tilde{\rho}_B^{x},
\label{eq:rho_AB_EC}
\end{equation}
and the performance index  (\ref{eq:QEC_def_with0}) can also regarded a function of $\rho_{X^AB}$,
\begin{equation}
Q^{{\rm EC},g}(\rho_{X^AB})=Q^{{\rm EC},g}(\rho_{X^AB}^0).
\label{eq:QEC_def}
\end{equation}

%revised 2021/10/07 from here 
\subsubsection{Relation with classical EC with quantum decoder in the past literature}
The EC defined above is a limitation of classical EC with quantum decoder, which has been discussed in various literatures (see e.g. Ref. \cite{Nielsen-Chuang}, Section 12.3 for the asymptotic case, and \cite{2018arXiv180911143C} for non-asymptotic cases).
The past literature and this paper are the same up to the point where the quantum channel receives codeword $c(t)$ and outputs a quantum state $\rho^{c(t)}$.
On the other hand, while $\rho^{c(t)}$ are general quantum state in the past literature, we here assume that $\rho^{c(t)}$ are cq states satisfying BSC condition (\ref{eq:BSC_condition}).

The main motivation for considering such restricted form of EC is its application to the security proof of QKD, whose details will be given in Section VI. A. 4.
To give a quick overview: In the phase error correction (PEC)-based approach of the QKD security proof, one can benefit by considering EC on a purification $\ket{\rho}_{ABE}$ (virtual state) of the actual state $\rho_{Z^AE}$, instead of dealing with $\rho_{Z^AE}$ directly.
This naturally gives rise to EC as defined above, where the additional information from the (virtual) ancilla space $B$ improves the performance of EC (and thus also improve the security of QKD).
%revised 2021/10/07 to here 

\subsection{Data compression (DC) with quantum side information}
%revised 2021/10/07 from here 
Similarly, a generalization of classical data compression (DC) is known, called {\it DC with quantum side information} (Fig. \ref{fig:data_compression}); see Ref. \cite{PhysRevA.68.042301} for the original definition, and Refs. \cite{doi:10.1098/rspa.2010.0445,TomamichelPhD,6157080,8047296,9261419} for recent results on one-shot and non-asymptotic cases.
From now on, whenever we say DC, we mean this generalized type.
%revised 2021/10/07 to here 

In a conventional classical DC algorithm \cite{10.5555/1146355}, given a classical data $x\in\{0,1\}^n$, one generates the corresponding compressed data $g(x)\in\{0,1\}^m$, using a compression function $g$, and stores it. 
After some time passes, one decodes $g(x)$ to restore $x$.

%revised 2021/10/07 from here 
On the other hand, in DC with quantum side information (Fig. \ref{fig:data_compression}) \cite{PhysRevA.68.042301,doi:10.1098/rspa.2010.0445,TomamichelPhD,8047296,6157080,9261419}, there is an additional set of sub-normalized quantum states (quantum side information) $\tilde{\rho}^x$ which correspond to $x$ and satisfy $\sum_x{\rm Tr}\tilde{\rho}^x=1$.
%revised 2021/10/07 to here 
The user can store $\tilde{\rho}^x$ along with the classical compressed data $g(x)$, and also use it for decoding. 
Thus, as in the EC algorithm of the previous subsection, it is possible that the decoding succeeds with a higher probability, due to information (i.e. hint) obtained by measuring $\tilde{\rho}^x$.
Details are as follows.

\subsubsection{Description of the algorithm}
For the sake of simplicity, we assume that the compression function $g:\{0,1\}^n\to\{0,1\}^m$ is linear.

\paragraph{Encoding} Choose variable $x$ with probability $\Pr[X^A=x]={\rm Tr}(\tilde{\rho}^{x})$.
Then output the compressed data $s=g(x)$ of $x$, along with the corresponding quantum side information $\tilde{\rho}^{x}$.

\paragraph{Decoding} Receive $s\in\{0,1\}^{n-m}$ and state $\tilde{\rho}_{B}^x$ as inputs.
Measure $\tilde{\rho}_{B}^x$ using a POVM $M^s=\{M^{s,e}\,|\,e\in\{0,1\}^n\}$, and output $x'=e$ as the estimated value of $x$.

\begin{figure}[t]
 \includegraphics[bb=0 0 950 500, width=\linewidth]{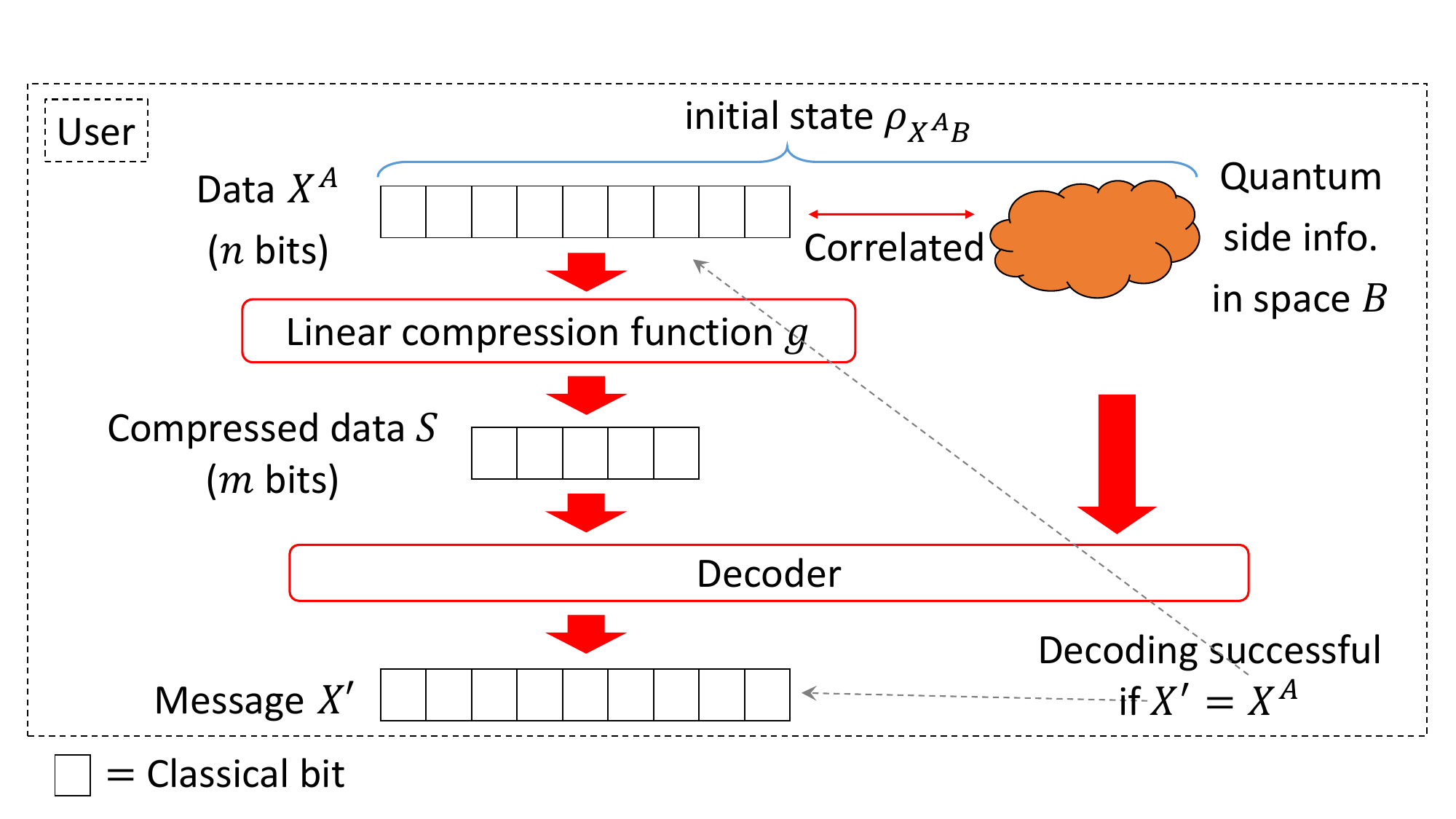}
 \caption{
Data compression (DC) with quantum side information.
Note that the situation becomes the same as in the conventional classical data compression using compression function $g$, if the quantum side information (orange cloud) is absent.
}
 \label{fig:data_compression}
\end{figure}

\subsubsection{Performance index}
We say that the decoding is successful when its output $x'$ equals the correct data $x$;
$\Pr[{\rm DC\ succeeds}]=\Pr[X'=X^A]$.
We then define the performance index of DC with quantum side information to be its failure probability when using the best decoder,
\begin{equation}
Q^{{\rm DC},g}(\{\tilde{\rho}^x\}):=\min_{\{M^s\}}\Pr[X'\ne X^A\,|\,\{\tilde{\rho}^x\}]
\label{eq:Q_DC_defined}
\end{equation}

If we particularly choose to encode the classical data $x$ in the $x$ basis, then the correlation between $x$ and $\tilde{\rho}^x$ can be represented by a single quantum state $\rho_{X^AB}=\sum_x\ket{\tilde{x}}\bra{\tilde{x}}\otimes\tilde{\rho}^x$, which has the identical form to the state $\rho_{X^AB}$ of (\ref{eq:rho_AB_EC}), which is used in EC.
That is, though $\rho_{X^AB}$ was originally introduced to describe the quantum channel of EC, it can also be regarded as describing the initial state of DC.
In this notation, the probability $Q^{{\rm DC},g}$ of (\ref{eq:Q_DC_defined}) can be considered as a function of $\rho_{X^AB}$, 
\begin{eqnarray}
Q^{{\rm DC},g}(\{\tilde{\rho}^x\})&=&Q^{{\rm DC},g}(\rho_{X^AB})\nonumber\\
&=&\min_{\{M^s\}}\Pr[X'\ne X^A\,|\,\rho_{X^AB}].
\label{eq:GDC_def}
\end{eqnarray}

\subsection{Note on the basis choices of the classical information}
\label{sec:note on the basis choices}

In the specification of algorithms above, for each algorithm we assigned different orthogonal bases for encoding the classical variable $z$ or $x$:
Variable $z$ in PA is encoded in the $z$ basis (without tilde), whereas variable $x$ in EC and DC is encoded in the $x$ basis (with tilde).

We note that we chose these particular bases solely for the purpose of simplifying the discussion of the following sections, where we prove the equivalence of the three algorithms.
We stress that we have no other reason for these particular choices.
For example, when one implements any one of the three algorithms in practice, one can use an arbitrary orthogonal basis for encoding classical variables, even besides the $z$ and the $x$ bases.

\section{Main results}
\label{sec:main_results}
Next we present the two main results of this paper.

The first result is that all three algorithms described above are actually equivalent.
That is, if any one of the three is specified, the other two can also be defined uniquely, and in addition, their indices, $Q^{\rm PA}$, $Q^{\rm EC}$, $Q^{\rm DC}$, are all equal.

The second result is that an leftover hashing lemma (LHL) for PA and coding theorems for EC and DC with quantum side information, are also equivalent.
That is, given either an LHL for PA, or a coding theorem for EC or DC, one can also derive the other two propositions (LHL or coding theorem) uniquely, and the three propositions thus obtained are all equivalent.

\subsection{Equivalence of the three algorithms}
\label{sec:equivalence_three_algorithms}

\subsubsection{Definition of the equivalence}
\label{sec:def_of_equivalence_three_algorithms}

By the equivalence of the three algorithms, we mean the following two conditions:
\begin{itemize}
\item [I.]{\bf Correspondence between the three types of algorithms:} If one specifies any one of the three algorithms (PA, EC and DC) and its input, the other two algorithms and their inputs are also specified uniquely and automatically.

Here, to ``specify an algorithm'' means to fix the function $f$ or $g$, which is used for the classical data processing\footnote{
It is straightforward to see that an PA algorithm can be specified uniquely by fixing the function $f$; and that an DC algorithm can be specified by $g$.
An EC algorithm can also be specified uniquely by the linear syndrome function $g$, as $g$ determines the corresponding linear error correcting code $C$, up to unessential linear transformations on input variable $x$}.
To ``specify the input to an algorithm'' is to fix the state $\rho_{Z^AE}$ or $\rho_{X^AB}$, which determines the initial state or the environment of the algorithm.

\item[II.] {\bf Equality of three indices:}
The indices of the three algorithms thus specified (i.e., $Q^{{\rm PA},f}(\rho_{Z^AE})$, $Q^{{\rm EC},g}(\rho_{X^AB})$, and $Q^{{\rm DC},g}(\rho_{X^AB})$) are all equal.
\end{itemize}

Condition I says that, once the hash function $f$ and the input $\rho_{Z^AE}$ for PA are fixed, then the function $g$ and the input $\rho_{X^AB}$ for EC (or for DC) are also fixed automatically; and vice versa.
This correspondence will be defined explicitly in Section \ref{sec:how_to_construct_three}.

Condition II says that once the analysis is finished for the security or the failure probability of any one of the three algorithms, then analyses for the other two algorithms become no longer necessary, since we know that they always give the same result.
From a practical viewpoint, this means, for example, that if one wishes to improve PA algorithms used for quantum cryptography, it suffices to improve DC or EC algorithms instead; and vice versa.
This equality will be shown in Section \ref{sec:subsec_equivalence_3_algorithms}.

In short, there is a correspondence between the three types of algorithms, and the triplet of algorithms connected by this correspondence share the same value of indices $Q$.
This means that the triplet of algorithms are in fact a single algorithm viewed from different angles.
Thus we call it the the equivalence of the three types of algorithms.

%The discussion below proceeds as follows.
%In Section \ref{sec:how_to_construct_three}, we begin by defining the correspondence (Condition 1) explicitly.
%There, we let linear functions $f$ and $g$ be {\it dual} to each other, and states $\rho_{Z^AE}$ and $\rho_{X^AB}$ be in what we call the {\it standard form}; cf. Fig. \ref{fig:construction_standard_state}.
%Then in Section \ref{sec:subsec_equivalence_3_algorithms}, we state that the equality of the index (Condition 2) indeed holds for algorithms thus specified, and prove it in Appendix \ref{sec:proof_of_theorem}.
%
%Subsequently in Section \ref{sec:randomizing_fg}, we point out that essentially the same results also hold when function $f$ and $g$ are randomized.
%In such randomized cases, it is straightforward to see that upper bounds on the averaged indices are in fact the LHL for PA, and the coding theorems for EC and DC.
%In Section \ref{sec:equivalence_LHL_conding_theorem}, by combining this observation and our equivalence, we show that the LHL for PA, and the coding theorems for EC and DC are also equivalent.

\subsubsection{Correspondence between the three types of algorithms (Condition I)}
\label{sec:how_to_construct_three}
Here we explicitly define the correspondence mentioned in Condition I above.
That is, we show that if one specifies the classical function (i.e., $f$ or $g$) and the input (i.e., $(\rho_{Z^AE}$ or $\rho_{X^AB}$) for any one of the three algorithms (PA, EC and DC), then the function and the input for the other two algorithms are also specified uniquely and automatically.
%More precisely, we show how to convert $\rho_{Z^AE}$ and function $f$, which specify the environment and the algorithm in PA, to their counterparts $\rho_{X^AB}$ and $g$ in EC and DC, and vice versa.

\paragraph{Correspondence between $f$ and $g$ (dual functions)}

Given either one of functions $f$ and $g$, one can choose the other function to be the {\it dual} function in the following sense.

\begin{Dfn}[Dual function pair \cite{6492260,7399404}]
\label{dfn:dual_functions}
Suppose a pair of linear functions $f:\{0,1\}^n\to\{0,1\}^m$ and $g:\{0,1\}^n\to\{0,1\}^{n-m}$ is given, and can be written $f(x)=fx^T$, $g(x)=gx^T$ using $m\times n$ matrices $f,g$.
Also suppose that functions $f,g$ are both surjective 
(i.e., matrices $f,g$ are of full rank).
We say that functions $f,g$ are dual, if the corresponding matrices $f,g$ satisfy $fg^T=0$.

We will often write $f\perp g$ or $g=f^\perp$ or $f=g^\perp$, to say that $f$, $g$ are dual.
\end{Dfn}

Note that we do not lose generality by restricting $f$ and $g$ to be  surjective.
A non-surjective $f$ or $g$ can always be modified to be surjective by discarding some of its output bits, and this modification does not affect the indices of the corresponding algorithm, $Q^{\rm PA}$, $Q^{\rm EC}$, and $Q^{\rm DC}$.

\paragraph{Correspondence between $\rho_{Z^AE}$ and $\rho_{X^AB}$ (standard form of tripartite states)}
\label{par:correspondence_states}
Given either one of states $\rho_{Z^AE}$ and $\rho_{X^AB}$, one can always define the other state uniquely by the following procedure (Fig. \ref{fig:construction_standard_state}).
\begin{enumerate}
\item Define a tripartite pure state $\rho_{ABE}$ by purifying $\rho_{Z^AE}$ or $\rho_{X^AB}$.

(E.g., if $\rho_{Z^AE}$ is given, introduce ancilla space $B$ and let $\rho_{ABE}$ such that $\rho_{Z^AE}={\rm Tr}_B(\ket{\rho}\bra{\rho}_{ABE})$.)
\item Trace out the unnecessary one of spaces $B$ and $E$.
\item Measure space $A$ in the appropriate basis, of the $z$ and $x$ bases.

(E.g., in order to obtain $\rho_{X^AB}$, perform the $X^A$-measurement on $\rho_{AB}={\rm Tr}_E(\ket{\rho}\bra{\rho}_{ABE})$.)
\end{enumerate}

In what follows, we will often refer to this procedure as $T$ \footnote{
In fact, applying $T$ twice in a row does not output the original state, e.g. $T(T(\rho_{Z^AE}))\ne \rho_{Z^AE}$, but this fact does not compromise our correspondence of the states.
This is because the state $T(T(\rho_{Z^AE}))$ and $\rho_{Z^AE}$ can be regarded equivalent as long as PA is concerned, and similarly,
$T(T(\rho_{X^AB}))$ and $\rho_{X^AB}$ are equivalent as long as DC or EC is concerned.
Details are given in Appendix \ref{sec:note_on_correspondence}.}.
The pure state $\rho_{ABE}$ obtained in this procedure satisfies $\rho_{AE}=\rho_{Z^AE}$ or $\rho_{AB}=\rho_{X^AB}$ by definition.
For later convenience, we call this type of $\rho_{ABE}$ the {\it standard form}.

\begin{Dfn}[Standard form of tripartite states; Fig. \ref{fig:construction_standard_state}]
\label{dfn:standard_form}
We say that a sub-normalized tripartite state $\rho_{ABE}$ is in the standard form, if it is pure, and satisfies either one of the conditions below, 
\begin{itemize}
\item $\rho_{AE}$ is classical in $Z^A$ (i.e., $\rho_{AE}=\rho_{Z^AE}$),
\item $\rho_{AB}$ is classical in $X^A$ (i.e., $\rho_{AB}=\rho_{X^AB}$).
\end{itemize}
\end{Dfn}

\begin{figure}[t]
 \includegraphics[bb=0 0 900 290, width=\linewidth]{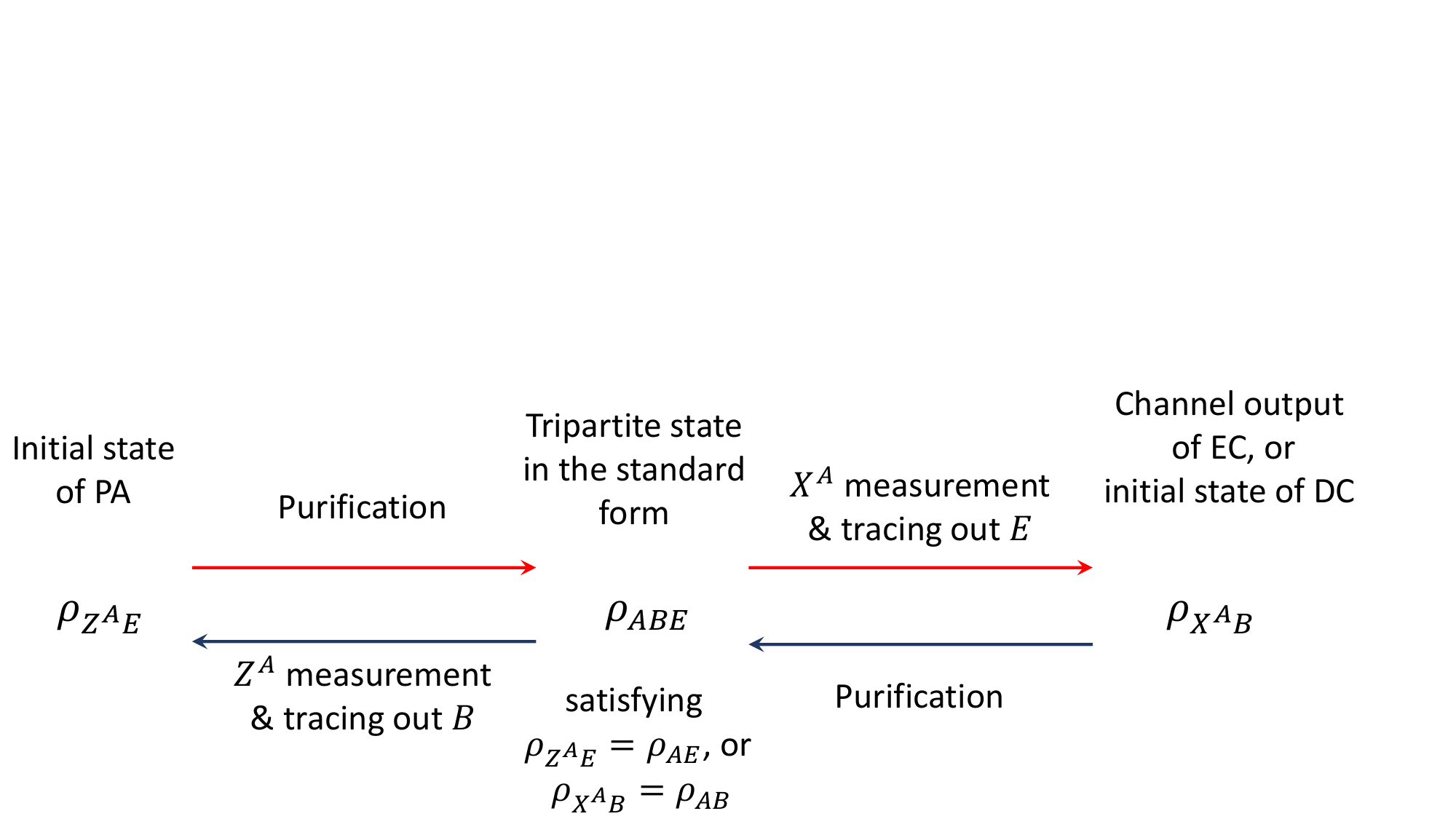}
 \caption{
Standard form of tripartite states:
The state $\rho_{Z^AE}$ that specifies the PA algorithm, and the state $\rho_{X^AB}$ that specifies EC and DC can be converted to each other uniquely, as shown above (cf. Section \ref{sec:how_to_construct_three}).
In these processes, one also obtains a pure tripartite state $\rho_{ABE}$ satisfying the condition $\rho_{AE}=\rho_{Z^AE}$ or $\rho_{AB}=\rho_{X^AB}$.
For later convenience, we call this type of $\rho_{ABE}$ the standard form (see Definition \ref{dfn:standard_form}).
}
 \label{fig:construction_standard_state}
\end{figure}

\subsubsection{Equality of three indices (Condition II)}
\label{sec:subsec_equivalence_3_algorithms}

Next we prove the equality of the indices of the three algorithms thus specified (Condition II of Section \ref{sec:def_of_equivalence_three_algorithms}).
We also show that they can all be expressed by the conditional min- and max-entropies.

\begin{Thm}[Equality of the indices; Fig. \ref{fig:tripartite}.]
\label{Thm:equivalence}
For a dual pair of functions, $f$ and $g$ $(=f^\perp)$, and for a sub-normalized state $\rho_{ABE}$ in the standard form,
\begin{eqnarray}
Q^{{\rm PA},f}(\rho_{Z^AE})&=&Q^{{\rm EC},g}(\rho_{X^AB})= Q^{{\rm DC},g}(\rho_{X^AB})\nonumber\\
&=&{\rm Tr}(\rho)-2^{H_{\rm max}(f(Z^A)|E)_{\rho}-m}\nonumber\\
&=&{\rm Tr}(\rho)-2^{-H_{\rm min}(X^A|B,g(X^A))_\rho}.
\label{eq:Thm_equality}
\end{eqnarray}
\end{Thm}
In other words, evaluating the security of PA is equivalent to evaluating the performances of EC and DC.
The proof of this theorem is given in Appendix \ref{sec:proof_of_theorem}.

% revised 2021/10/05 from here
We note that Renes has already shown part of Theorem \ref{Thm:equivalence}, the equivalence of PA and DC (Ref. \cite{8047296}, Corollary 11).
Thus our contribution in Theorem \ref{Thm:equivalence} is that we extended the equivalence to also include EC.

%The difference between Theorem \ref{Thm:equivalence} and the results in the past literature is as follows.
%In Ref. \cite{doi:10.1098/rspa.2010.0445}, Renes also discussed the relation between PA and DC in a similar setting, but his bounds were not tight enough to prove their equivalence.
%As a proof of this observation, we note that the bound $p_{\rm guess}(Z|B)_\psi\ge1-\sqrt{2\epsilon}$ in Theorem 4 of Ref. \cite{doi:10.1098/rspa.2010.0445} can be improved to $p_{\rm guess}(Z|B)_\psi\ge1-2\epsilon$ by combining our Theorem \ref{Thm:equivalence} and Eq. (9.110) of Ref. \cite{Nielsen-Chuang}.

% revised 2021/10/05 to here

\begin{figure}[t]
 \includegraphics[bb=0 0 950 500, width=\linewidth]{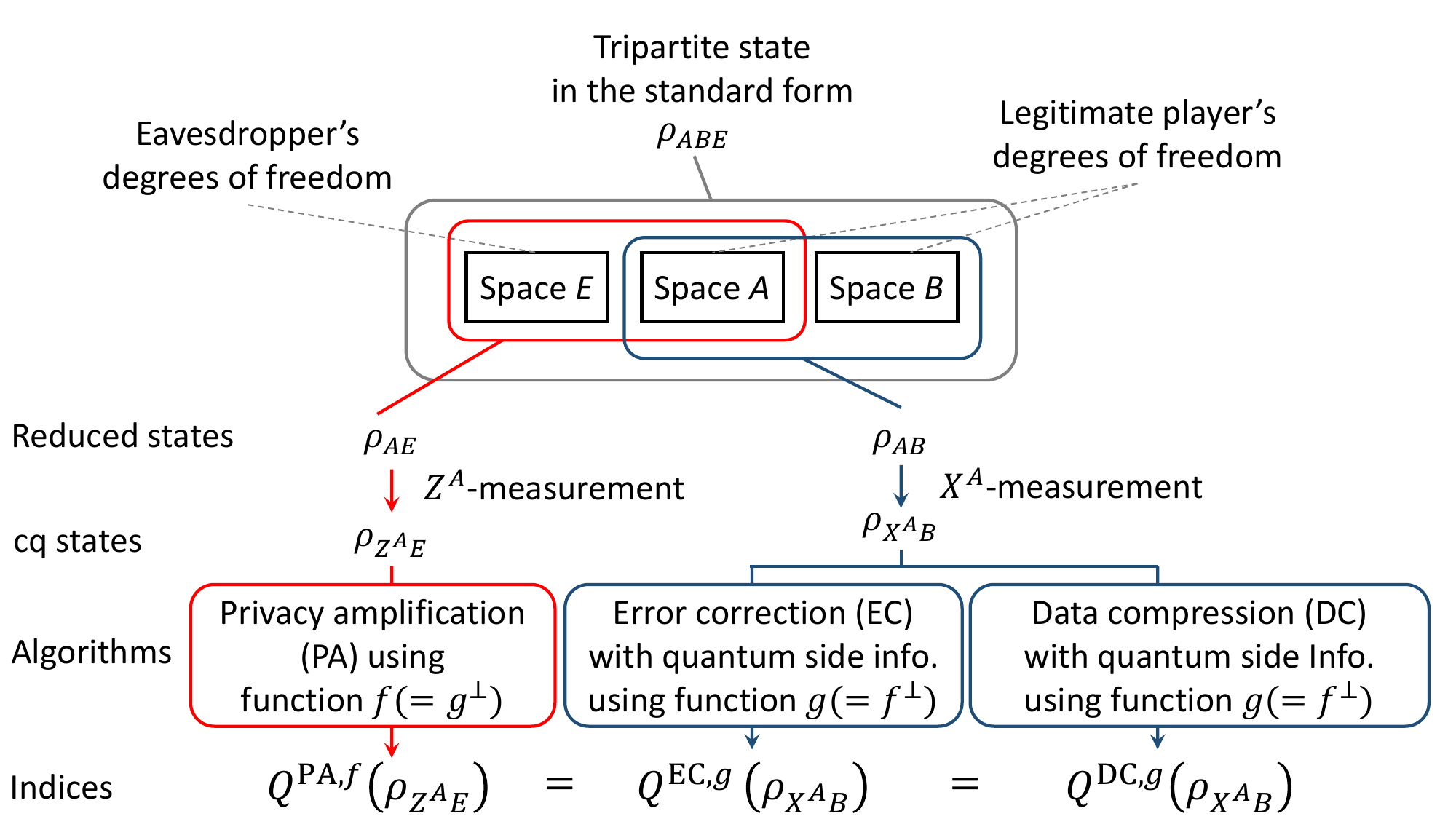}
 \caption{
Equality of three indices:
The indices $Q^{\rm PA}$, $Q^{\rm EC}$, $Q^{\rm DC}$ of the three algorithms are equal (Section \ref{sec:subsec_equivalence_3_algorithms}), if their inputs $\rho_{Z^AE}$ and $\rho_{X^AE}$ are related via tripartite state $\rho_{ABE}$ in the standard form (cf. Fig. \ref{fig:construction_standard_state}), and if functions $f$ and $g$ are dual (cf. Definition \ref{dfn:dual_functions}). }
 \label{fig:tripartite}
\end{figure}

\subsection{Random functions $F,G$}
\label{sec:randomizing_fg}
Next suppose that one chooses function $f$ (or $g$) randomly from a given set ${\cal F}$ (or ${\cal G}$) with a given probability $p(f)$ (or $q(g)$), every time one executes the algorithm.
This corresponds to the case where one uses a random hash function in PA, or a random code in EC and DC.

Even in this setting, Theorem \ref{Thm:equivalence} is true, as long as $f$ and $g$ are randomized in such a way that their duality is maintained.
Details are as follows. 

In order to simplify the notation of these random functions $f,g$, from now on, we will consider them as random variables and write them in uppercase $F,G$.
In this notation, the occurrence probability of $f$, for example, is denoted $\Pr[F=f]=p(f)$, and the expected value of $r(f)$, a function $r$ of $f$, is ${\rm E}_F\,r(F)=\sum_{f\in{\cal F}}p(f)r(f)$.
The same is true for function $g$.

Then we say that a pair of random functions $F,G$ is dual, if they are chosen randomly with their duality (in the sense of Definition \ref{dfn:dual_functions}) maintained; that is,
\begin{Dfn}
We say that a pair of random functions $F,G$ is dual, if they are chosen from sets of the same size, ${\cal F}=\{f_1,f_2,\dots\}$, ${\cal G}=\{g_1,g_2,\dots\}$, respectively, and if each function pair $f_i, g_i$ is dual (in the sense of Definition \ref{dfn:dual_functions}) and is chosen with the same probability, i.e.,
\begin{equation}
f_i\perp g_i\ {\rm and}\ \Pr[F=f_i]=\Pr[G=g_i]\ {\rm for}\ \forall i.  
\end{equation}

We also write $F\perp G$ or $G=F^\perp$ or $F=G^\perp$, if $F$, $G$ are dual.
\end{Dfn}

By using this notation, we can state an averaged version of Theorem \ref{Thm:equivalence} as follows.
%\begin{Crl}
%\label{crl:Q_PAF_averaged_inequality}
%Let $F,G$ be dual random functions, and $\rho_{ABE}$ be any state, then 
%\begin{equation}
%{\rm E}_F\, Q^{{\rm PA},F}(\rho)\le {\rm E}_{G}\, Q^{{\rm EC},G}(\rho)={\rm E}_{G}\, Q^{{\rm DC},G}(\rho).
%\end{equation}
%\end{Crl}
\begin{Crl}
\label{crl:equivalence_random_case}
Let $F,G$ be a pair of dual random functions, and $\rho_{ABE}$ be a sub-normalized state in the standard form, then
\begin{eqnarray}
\lefteqn{{\rm E}_F\, Q^{{\rm PA},F}(\rho_{Z^AE})}\nonumber\\
&= &{\rm E}_{G}\, Q^{{\rm EC},G}(\rho_{X^AB})={\rm E}_{G}\, Q^{{\rm DC},G}(\rho_{X^AB})\nonumber\\
&=&{\rm Tr}(\rho)-{\rm E}_F\, 2^{H_{\rm max}(F(Z^A)|E)_{\rho}-m}\nonumber\\
&=&{\rm Tr}(\rho)-{\rm E}_{G}\, 2^{-H_{\rm min}(X^A|B,G(X^A))_\rho}.
\label{eq:equivalence_random_case}
\end{eqnarray}
\end{Crl}

\subsection{Equivalence of leftover hashing lemmas (LHL) and coding theorems}
\label{sec:equivalence_LHL_conding_theorem}

\subsubsection{Equivalence of LHLs and coding theorems}

If the average security index ${\rm E}_F\, Q^{{\rm PA},F}(\rho_{Z^AE})$, appearing in Corollary \ref{crl:equivalence_random_case}, can be bounded by some function of the initial state $\rho_{Z^AE}$, it will be called a leftover hashing lemma (LHL).
Similarly, if the average failure probability ${\rm E}_{G}\, Q^{{\rm EC},G}(\rho_{X^AB})$ of EC (or ${\rm E}_{G}\, Q^{{\rm DC},G}(\rho_{X^AB})$ of DC) can be bounded by some function of the environment state $\rho_{X^AE}$, it will be called a coding theorem.

We can show that these LHLs and coding theorems are equivalent, if we combine Corollary \ref{crl:equivalence_random_case} and the duality of min- and max-entropies (Lemma \ref{lmm:uncertainty} below).
\begin{Thm}[Equivalence of an LHL and coding theorems]
\label{thm:equivalence_LHL_coding}
Let $F,G$ be a pair of dual random functions, $r$ be an arbitrary function, and $\rho_{Z^AE}$ and $\rho_{X^AB}$ be normalized states.
Then the following three inequalities are equivalent; that is, if any one of them is true, then all of them are true.
\begin{itemize}
\item An LHL for PA using $F$ ($=G^\perp$),
\begin{equation}
{\rm E}_F\, Q^{{\rm PA},F}(\rho_{Z^AE})\le r(H_{\rm min}(Z^A|E)_\rho).
\label{eq:thm:LHL-like_bound}
\end{equation}
\item Coding theorem for EC using $G$ ($=F^\perp$),
\begin{equation}
{\rm E}_{G}\, Q^{{\rm EC},G}(\rho_{X^AB})\le r(n-H_{\rm max}(X^A|B)_\rho).
\label{eq:thm:coding_theorem_EC}
\end{equation}
\item Coding theorem for DC using $G$ ($=F^\perp$),
\begin{equation}
{\rm E}_{G}\, Q^{{\rm DC},G}(\rho_{X^AB})\le r(n-H_{\rm max}(X^A|B)_\rho).
\label{eq:thm:coding_theorem_DC}
\end{equation}
\end{itemize}
\end{Thm}

This theorem is a direct consequence of Corollary \ref{crl:equivalence_random_case} above, and Lemma \ref{lmm:uncertainty} below.
The left hand sides of (\ref{eq:thm:LHL-like_bound}), (\ref{eq:thm:coding_theorem_EC}), (\ref{eq:thm:coding_theorem_DC}) are equal due to Corollary \ref{crl:equivalence_random_case}, and the right hands are also equal due to Lemma \ref{lmm:uncertainty}.
%Recall here that, as mentioned in Section \ref{sec:equivalence_three_algorithms} and in Figure \ref{fig:construction_standard_state}, when given either of two states $\rho_{Z^AE}$ and $\rho_{X^AB}$, one can always construct the other, along with the corresponding standard form $\rho_{ABE}$.

\begin{Lmm}[Entropic uncertainty relation \cite{PhysRevLett.106.110506}, and its equality condition \cite{PhysRevLett.108.210405}]
\label{lmm:uncertainty}
Let $A$ be an $n$-qubit space and $\rho_{ABE}$ be a sub-normalized pure state.
Then for $\varepsilon\ge0$, we have \cite{PhysRevLett.106.110506}
\begin{equation}
H_{\rm max}^\varepsilon(X^A|B)_\rho+H_{\rm min}^\varepsilon(Z^A|E)_{\rho}\ge n.
\label{eq:UCR}
\end{equation}
The inequality holds if $\rho_{ABE}$ is in the standard form \cite{PhysRevLett.108.210405}.
\end{Lmm}

\subsubsection{Equivalence of the generalized bounds}

We note that Theorem \ref{thm:equivalence_LHL_coding} above can also be generalized in a straightforward way.

In Theorem \ref{thm:equivalence_LHL_coding}, we implicitly assumed that the average security index ${\rm E}_F\, Q^{{\rm PA},F}(\rho_{Z^AE})$ of PA should be bounded by a function of its initial state $\rho_{Z^AE}$, and also that the failure probability ${\rm E}_{G}\, Q^{{\rm EC},G}(\rho_{X^AB})$ of EC (or ${\rm E}_{G}\, Q^{{\rm DC},G}(\rho_{X^AB})$ of DC) should be bounded by its environment state $\rho_{X^AB}$.
However, as one can see from the proof of Theorem \ref{thm:equivalence_LHL_coding}, such restriction is not in fact essential in our formalism. 
Rather, we may consider bounds of the generalized forms by removing such restriction, and prove their equivalence.
%For example, one can bound the security index ${\rm E}_F\, Q^{{\rm PA},F}(\rho_{Z^AE})$ of PA by a function of the environment state $\rho_{X^AB}$ of EC.

\begin{Thm}[Generalization of Theorem \ref{thm:equivalence_LHL_coding}]
\label{Thm:H_min_H_max}
Let $F,G$ be a pair of dual random functions, $r$ be a function, and $\rho_{ABE}$ be a normalized state in the standard form.
Then the six inequalities of the form
\begin{equation}
a\le r(b)
\label{eq:general_inequality_form}
\end{equation}
are all equivalent, 
where 
\begin{eqnarray*}
a&\in&\{{\rm E}_F\, Q^{{\rm PA},F}(\rho_{Z^AE}),\ {\rm E}_{G}\,Q^{{\rm EC},G}(\rho_{X^AB}),\\
&&\quad {\rm E}_{G}\, Q^{{\rm DC},G}(\rho_{X^AB})\},\\
b&\in&\{H_{\rm min}(Z^A|E)_\rho,\ n-H_{\rm max}(X^A|B)_\rho\}.
\end{eqnarray*}
That is, if one inequality of the form (\ref{eq:general_inequality_form}) is true, then all the six are true.
\end{Thm}

%These inequality are useful, for example, when considering the security of quantum key distribution (QKD) in later sections.

\section{Application 1: Leftover Hashing Lemmas for (dual) universal$_2$ functions $F$, and the corresponding coding theorems}
\label{sec:application1_LHL_coding_theorem}

Next we apply our main results of the previous section, Theorem \ref{Thm:H_min_H_max} in particular, to explicit examples of dual random function pair, $F$ and $G$ $(=F^\perp)$.

Specifically, we choose $F$ be the three classes of hash functions that are commonly used for PA: the universal$_2$ \cite{CARTER1979143}, the almost universal$_2$ \cite{WEGMAN1981265}, and the dual universal$_2$ functions \cite{FS08,6492260}.
Then we prove LHLs of the new type (i.e. bounds on ${\rm E}_F\,Q^{{\rm PA},F}$; cf. \ref{sec:LHL_new_type}) for these $F$.

Further, we apply Theorem \ref{Thm:H_min_H_max} to the new LHLs thus obtained, and demonstrate that the coding theorems for EC and DC using the dual random function $F^\perp$ follows automatically.

To the best of our knowledge, most of the LHLs and the coding theorems obtained below are new results that have not been obtained previously (also see the first paragraph of Section \ref{sec:explicit_example_LHL}).

\subsection{Basic strategy}
The basic strategy here is to borrow previous results related with LHLs of the conventional form (i.e. bounds on the standard security index ${\rm E}_F\,d_1(\rho^F_{KE})$) \cite{RennerPhD,5961850,FS08,6492260}, and convert them to LHLs of the new type.
The actual procedure is as follows.

\subsubsection{Previous derivation of LHLs of the conventional form}
First we review the previous results of Refs. \cite{RennerPhD,5961850,FS08,6492260}.
In these papers, LHLs of the conventional form were obtained by the following two steps.
\begin{enumerate}
\item Prove a bound of the form
\begin{equation}
{\rm E}_F\, d_2(\rho_{KE}^F|\sigma_E)\le 2^{-m}r(H_2(\rho_{Z^AE}|\sigma_E)),
\label{eq:d_2_bound}
\end{equation}
where $r(\cdot)$ is a non-increasing function, and $d_2(\cdot)$ and $H_2(\cdot|\cdot)$ are defined by
\begin{align}
&d_2(\rho_{KE}|\sigma_E)\nonumber\\
&:={\rm Tr}
\left\{
\left(
\left(\rho_{KB}-2^{-m}\II_K\otimes\rho_E\right)\left(\II_K\otimes\sigma_E^{-1/2}\right)
\right)^2
\right\},\\
&H_2(\rho_{Z^AE}|\sigma_E):=-\log {\rm Tr}\left\{\left(\rho_{Z^AE}\left(\II_A\otimes \sigma_E^{-1/2}\right)\right)^2\right\}
\end{align}
for a normalized $\sigma$ (see, e.g., \cite{RennerPhD}).

\item By using relations
\begin{eqnarray}
d_1(\rho^f_{KE})&\le& \sqrt{2^md_2(\rho^f_{KE}|\sigma_E)},
\label{eq:d_1_bounded_by_d2}\\
H_{\rm min}(Z^A|E)_{\rho}&\le&H_2(\rho_{Z^AE}|\sigma_E)
\label{eq:H_min_bounded_by_H_2} 
\end{eqnarray}
which hold for a normalized $\sigma$ (see Ref. \cite{RennerPhD} for the proofs), and by using Jensen's inequality, one obtains an LHL of the conventional form,
\begin{equation}
{\rm E}_F\,d_1(\rho^F_{KE})\le \sqrt{r(H_{\rm min}(Z^A|E)_{\rho})}.
\label{eq:conventional_LHL}
\end{equation}
\end{enumerate}

\subsubsection{How to derive the new type of LHLs}
\label{sec;how_to_derive_LHL_QPA}

As we have seen above, bounds of the form (\ref{eq:d_2_bound}) have already been proved for some random functions $F$ \cite{RennerPhD,5961850,FS08,6492260}.
These bounds can readily be converted to the new type of LHLs (i.e. bounds on ${\rm E}_F\, Q^{{\rm PA},F}(\rho_{Z^AE})$) by using the following theorem.

\begin{Thm}
\label{thm:d_2_bound}
Let $\rho_{Z^AE}$ be a normalized state, $F$ be a random function, and $r$ be a non-increasing function.
If we have a bound of the form
\begin{equation}
{\rm E}_F\, d_2(\rho_{KE}^F|\rho_E)\le 2^{-m}r(H_2(\rho_{Z^AE}|\rho_E))
\end{equation}
(i.e., Inequality (\ref{eq:d_2_bound}) with $\sigma_E=\rho_E$),
then we have an LHL for PA using $F$,
\begin{equation}
{\rm E}_F\, Q^{{\rm PA},F}(\rho_{Z^AE})\le r(H_{\min}(Z^A|E)_\rho),
\label{eq:EFQPA_for_general_F}
\end{equation}
and coding theorems for EC and DC using $G=F^\perp$,
\begin{align}
&{\rm E}_F\, Q^{{\rm EC},F^\perp}(\rho_{X^AB})={\rm E}_F\, Q^{{\rm DC},F^\perp}(\rho_{X^AB})\nonumber\\
&\quad\le r(n-H_{\max}(X^A|B)_\rho),
\end{align}
which holds for a normalized state $\rho_{X^AB}$.
\end{Thm}

Note here that if we apply (\ref{eq:d1_Q_bound}) and Jensen's inequality to Inequality (\ref{eq:EFQPA_for_general_F}), we again derive an LHL of the conventional type
\begin{equation}
{\rm E}_F\, d_1(\rho_{KE})\le 4 \sqrt{r(H_{\min}(Z^A|E)_\rho)}.
\end{equation}
Also note that this inequality differs only by the factor of four from the previous result (\ref{eq:conventional_LHL}), which was obtained without using (\ref{eq:EFQPA_for_general_F}).
In this sense we say that the new bound (\ref{eq:EFQPA_for_general_F}) is nearly as tight as the previous bound (\ref{eq:conventional_LHL}).

\subsubsection{Proof of Theorem \ref{thm:d_2_bound}}
Theorem \ref{thm:d_2_bound} can be proved by using Theorem \ref{thm:equivalence_LHL_coding} and the following two lemmas.
\begin{Lmm}
For any random hash function $F$, and for a normalized state $\rho_{Z^AE}$,
\begin{equation}
Q^{{\rm PA},f}(\rho_{Z^AE})\le 2^m d_2(\rho_{KE}^f|\rho_E).
\end{equation}
\end{Lmm}
\begin{proof}
According to Ref. \cite{TBH14}, we have
\begin{equation}
\widetilde{H}_{1/2}^\downarrow(K|E)_\rho\ge \widetilde{H}_{2}^\downarrow(K|E)_\rho,
\end{equation}
where
\begin{align}
\widetilde{H}_{1/2}^\downarrow(K|E)_\rho&:=2\log_2 F(\II_K\otimes\rho_E,\rho_{KE}),\\
\widetilde{H}_{2}^\downarrow(K|E)_\rho&:=-\log_2 {\rm Tr}\left\{\left(\rho_{KE}\left(\II_K\otimes \rho_E^{-1/2}\right)\right)^2\right\}.
\label{eq:tilde_H2_down_defined}
\end{align}
Also by definition of $H_{\rm max}(K|E)_\rho$ and $\widetilde{H}_{1/2}^\downarrow(K|E)_\rho$, we have $H_{\rm max}(K|E)_\rho\ge \widetilde{H}_{1/2}^\downarrow(K|E)_\rho$.
Thus
\begin{eqnarray}
Q^{{\rm PA},f}(\rho)&\le&1-2^{-m+\widetilde{H}_{1/2}^\downarrow(K|E)_\rho}\nonumber
\\
&\le&1-2^{-m+\widetilde{H}_{2}^\downarrow(K|E)_\rho}\le2^{m-\widetilde{H}_{2}^\downarrow(K|E)_\rho}-1\nonumber\\
&=&2^m{\rm Tr}\left\{\left(\rho_{KE}\left(\II_K\otimes \rho_E^{-1/2}\right)\right)^2\right\}-1\nonumber\\
&=&2^m d_2(\rho_{KE}|\rho_E),
\label{eq:lemma2:1-2mF}
\end{eqnarray}
where in the second line we used $1-1/x\le x-1$ for $ x>0$.
\end{proof}
\begin{Lmm}
For any random hash function $F$, and for a normalized state $\rho_{Z^AE}$,
\begin{equation}
H_{\min}(Z^A|E)_\rho\le H_2(\rho_{Z^AE}|\rho_E).
\label{eq:lemma2:H_Alpha_relation}
\end{equation}
\end{Lmm}
\begin{proof}
$H_2(\rho_{Z^AE}|\rho_E)$ equals the quantity $\widetilde{H}_{2}^\downarrow(Z^A|E)_{\rho}$ defined in Ref. \cite{TBH14}, 
and $H_{\min}(Z^A|E)_\rho$ equals $\widetilde{H}_{\infty}^\uparrow(Z^A|E)_{\rho}$ also defined in the same paper.
Hence it suffices to show
\begin{equation}
\widetilde{H}_{\infty}^\uparrow(Z^A|E)_{\rho}\le\widetilde{H}_{2}^\downarrow(Z^A|E)_{\rho}.
\label{eq:lemma2:H_Alpha_relation}
\end{equation}
This inequality follows by substituting $\alpha=\infty$ in Eq. (48), Corollary 4, Ref. \cite{TBH14}.
\end{proof}

%By combining (\ref{eq:proof_lemma2_EFTr_rho}) and (\ref{eq:lemma2:H_Alpha_relation}), we obtain the lemma.

\subsection{Explicit examples of an LHL and the corresponding coding theorems}
\label{sec:explicit_example_LHL}

Next we apply Theorem \ref{thm:d_2_bound} above to three examples of random functions $F$.
To the best of our knowledge, all lemmas in this subsection, except Lemma \ref{lmm:LHL1}, are new results that have not been obtained previously.

\subsubsection{Case where $F$ is universal$_2$}
The first example of $F$ is the universal$_2$ function.
This is the most commonly used random functions for PA.

\begin{Dfn}[Linear universal$_2$ function \cite{CARTER1979143}]
\label{dfn:universal2_defined}
A linear random function $F:\{0,1\}^n\to\{0,1\}^m$ is universal$_2$ if
\begin{equation}
\Pr[F(x)=0]\le 2^{-m}\ {\rm for}\ \forall x\ne0.
\label{eq:universal2_def}
\end{equation}
\end{Dfn}

\paragraph{Leftover hashing lemma}
For this type of $F$, an inequality of the type of (\ref{eq:d_2_bound}) is already known, which takes the form
\begin{equation}
{\rm E}_F\, d_2(\rho_{KE}|\sigma_E)\le2^{-H_2(\rho_{Z^AE}|\sigma_E)}
\label{eq:proof_lemma2_EFTr_rho}
\end{equation}
for a normalized $\rho_{Z^AE}$ (Lemma 5.4.3, Ref. \cite{RennerPhD}).
If one applies relations (\ref{eq:d_1_bounded_by_d2}) and (\ref{eq:H_min_bounded_by_H_2}) to this inequality, one obtains an LHL of the conventional form,
\begin{equation}
{\rm E}_F\, d_1(\rho_{KE}^F)\le \sqrt{{\rm Tr}(\rho)}\sqrt{2^{m-H_{\rm min}(Z^A|E)_\rho}}
\label{eq:original_LHL}
\end{equation}
for a sub-normalized $\rho_{Z^AE}$ (Theorem 5.5.1 of Ref. \cite{RennerPhD}).

On the other hand, if we instead apply the former half of Theorem \ref{thm:d_2_bound} to (\ref{eq:proof_lemma2_EFTr_rho}), we obtain an LHL of the new type.
\begin{Lmm}[LHL for PA using a universal$_2$ hash function, in terms of $Q^{\rm PA}$ (Lemma 12, Ref. \cite{8970489})]
\label{lmm:LHL1}
For a universal$_2$ function $F$, and for a sub-normalized state $\rho_{Z^AE}$,
\begin{equation}
{\rm E}_F\, Q^{{\rm PA},F}(\rho_{Z^AE})\le 2^{m-H_{\rm min}(Z^A|E)_\rho}.
\label{eq:bounds_universal2}
\end{equation}
\end{Lmm}
As noted in Section \ref{sec:relation_with_the_conventional}, this LHL of the new type is nearly as tight as the LHL (\ref{eq:original_LHL}) of the conventional type, which has been obtained previously for the same type of $F$.
%Indeed  and also in Section \ref{sec;how_to_derive_LHL_QPA}, by applying (\ref{eq:d1_Q_bound}) and Jensen's inequality to (\ref{eq:bounds_universal2}), we can obtain Inequality Eq. (\ref{eq:LHL_universal_2}), which differs from the conventional LHL only by the factor of four.
%We regard this extra factor of two inessential because it can easily be compensated for by shortening the random number length $m$ by two bits, while $m>100$ for typical applications.

\paragraph{Coding theorems}
On the other hand, if we apply Theorem \ref{thm:equivalence_LHL_coding} to Lemma \ref{lmm:LHL1}, we readily obtain coding theorems for EC and DC using $G$ $(=F^\perp)$, with $F$ being universal$_2$.
\begin{Lmm}[Coding theorems for EC and DC using a dual universal$_2$ function (Lemma 12, Ref. \cite{8970489})]
\label{lmm:coding_theorems_dual_univ}
Let $G$ be a random function whose dual $G^\perp$ is universal$_2$
(or equivalently, let $G$ be a dual universal$_2$ function; see Definition \ref{def:dual_universal}).
Then for a sub-normalized state $\rho_{X^AB}$,
\begin{eqnarray}
{\rm E}_{G}\, Q^{{\rm EC},G}(\rho_{X^AB})&=&{\rm E}_{G}\, Q^{{\rm DC},G}(\rho_{X^AB})\nonumber\\
&\le& 2^{H_{\rm max}(X^A|B)_\rho-(n-m)}.
\end{eqnarray}
\end{Lmm}

%In Refs. \cite{6492260,7399404}, we used to call the type of random functions $G$ specified above, the {\it dual universal$_2$ function}, for an obvious reason.
%Hence the specification of $G$ in Lemma \ref{lmm:coding_theorems_dual_univ} can be rephrased as ``Let $G$ be an dual universal$_2$ function.''
%Note that such $G$ can also be specified by a condition on $G$ only, without referring to $F$, as $\Pr[x\in{\rm ker}\,G]\le 2^{-m}\ {\rm for}\ \forall x\ne0$,

\subsubsection{Case where $F$ is almost universal$_2$}

The second example of $F$ is a generalization of the universal$_2$ function above, which is defined as follows. 
\begin{Dfn}[Linear almost universal$_2$ function \cite{WEGMAN1981265}]
\label{dfn:almost_universal}
Let $\delta\in\RR$ be a parameter.
A linear random function $F:\{0,1\}^n\to\{0,1\}^m$ is $\delta$-almost universal$_2$ if
\begin{equation}
\Pr[F(x)=0]\le 2^{-m}\delta\ {\rm for}\ \forall x\ne0.
\end{equation}
\end{Dfn}
We note that $\delta$ must satisfy $(2^n-2^m)/(2^n-1)\le\delta$, as can be shown by  a simple argument \cite{WEGMAN1981265}.
Note that the universal$_2$ case of the previous subsection corresponds to $\delta=1$, near the minimum value.

\paragraph{Leftover hashing lemma}
For this type of $F$ as well, an inequality of the type of (\ref{eq:d_2_bound}) is already known,
\begin{equation}
{\rm E}_F\, d_2(\rho_{KE}^F|\rho_E)\le 2^{-m}(\delta-1)+2^{-H_2(\rho_{Z^AE}|\rho_E)}
\label{eq:bound_d2_type_almost_universal}
\end{equation}
for a normalized $\rho_{Z^AE}$ (Lemma 5 of Ref. \cite{5961850}), as well as an LHL of the conventional type
\begin{eqnarray}
\lefteqn{{\rm E}_F\, d_1(\rho_{KE}^F)}
\label{eq:LHL_conventional_almost_universal}
\\
&\le& \inf_{\varepsilon>0}\frac12\sqrt{ \delta-1+2^{m-H_2(\rho_{Z^AE}|\rho_E)+\log(2/\varepsilon^2+1)}}+\varepsilon\nonumber
\end{eqnarray}
(Lemma 2 of Ref. \cite{5961850}).

By applying the former half of Theorem \ref{thm:d_2_bound} to (\ref{eq:bound_d2_type_almost_universal}), we obtain an LHL of the new type.
\begin{Lmm}[LHL for PA using an almost universal$_2$ hash function, in terms of $Q^{\rm PA}$]
\label{lmm:LHL_QPA_almost_universal}
For a universal$_2$ hash function $F$ and for a sub-normalized state $\rho_{Z^AE}$,
\begin{equation}
{\rm E}_F\, Q^{{\rm PA},F}(\rho_{Z^AE})\le (\delta-1){\rm Tr}(\rho)+ 2^{m-H_{\rm min}(Z^A|E)_\rho}.
\label{eq:EFQ_almost_univ}
\end{equation}
\end{Lmm}

Note that if we set $\delta=1$, we again obtain Lemma \ref{lmm:LHL1}.
Also note that the bound (\ref{eq:EFQ_almost_univ}) loses its meaning for $2< \delta$, where its right hand exceeds one.
Therefore, PA using this type of $F$ can only be secure when $\delta$ is confined within a relatively small region, $1\simeq(2^n-2^m)/(2^n-1)\le \delta<2$ (see Section VIII.B of Ref. \cite{6492260} for a more detailed argument).
This seems to reflect the fact that this type of $F$ was originally conceived for information theoretically secure authentication \cite{WEGMAN1981265}, and not for PA.

\paragraph{Coding theorems}
As in the previous subsection, if we apply Theorem \ref{thm:equivalence_LHL_coding} to Lemma \ref{lmm:LHL_QPA_almost_universal}, we readily obtain coding theorems for EC and DC using a dual random function $G$ $(=F^\perp)$.
\begin{Lmm}[Coding theorems for EC and DC using an almost dual universal$_2$ function]
\label{lmm:coding_theorem_almost_dual_univ}
Let $G$ be a random function whose dual function $G^\perp$ is $\delta$-almost universal$_2$
(or equivalently, let $G$ be a $\delta$-almost dual universal$_2$ function; see Definition \ref{def:dual_universal}).
Then for a sub-normalized state $\rho_{X^AB}$,
\begin{eqnarray}
\lefteqn{{\rm E}_{G}\, Q^{{\rm EC},G}(\rho_{X^AB})={\rm E}_{G}\, Q^{{\rm DC},G}(\rho_{X^AB})}\nonumber\\
&\le& (\delta-1){\rm Tr}(\rho) + 2^{H_{\rm max}(X^A|B)_\rho-(n-m)}.
\end{eqnarray}
\end{Lmm}

\subsubsection{Case where $G$ is almost universal$_2$, or equivalently, where $F$ is almost dual universal$_2$}
In the third example of random functions, we switch properties of $F$ and $G$, and let $G$ $(=F^\perp)$ be an almost universal$_2$ function, instead of $F$.
In this case, $F$ is said to be an almost {\it dual} universal$_2$ function.
\begin{Dfn}[Almost dual universal$_2$ function \cite{FS08,6492260}]
\label{def:dual_universal}
A random function $F$ is called $\delta$-almost dual universal$_2$, if its dual random function $F^\perp$ is $\delta$-almost universal$_2$.
\end{Dfn}
%This type of $F$ can alternatively be defined by itself, without referring to its dual $G=F^\perp$, by
%\begin{equation}
%\Pr[x\in (\ker\, F)^\perp]\le 2^{-(n-m)}\delta\ {\rm for}\ \forall x\ne0,
%\end{equation}
%where $\ker\,f:=\{x\in\{0,1\}^n\,|\,f(x)=0\}$ is the kernel of a function $f$, and $(\ker\,f)^\perp$ is its dual vector space.

We note that for this type of $F$, parameter $\delta$ must satisfy $(2^n-2^{n-m})/(2^n-1)\le\delta$ \cite{6492260}.
We also note that this type of $F$ is $2(1-2^{-m}\delta)+(\delta-1)2^{n-m}$-almost universal$_2$ simultaneously \cite{6492260}\footnote{Conversely, the type of random function $F$ that is defined in Definition \ref{dfn:almost_universal} is also $2(1-2^{-m}\delta)+(\delta-1)2^{n-m}$-almost {\it dual} universal$_2$ \cite{6492260}.}.

\paragraph{Leftover hashing lemma}
For this type of $F$ as well, an inequality of the type of (\ref{eq:d_2_bound}) is known \cite{FS08,6492260},
\begin{equation}
{\rm E}_F\, d_2(\rho_{KE}^F|\sigma_E)\le 2^{-H_2(\rho_{Z^AE}|\sigma_E)}\delta,
\label{eq:d_2_bound_almost_universal}
\end{equation}
for a normalized $\rho_{Z^AE}$.

If one applies relations (\ref{eq:d_1_bounded_by_d2}) and (\ref{eq:H_min_bounded_by_H_2}) to this inequality, one obtains an LHL of the conventional form,
\begin{equation}
{\rm E}_F\, d_1(\rho_{KE}^F)\le \sqrt{{\rm Tr}(\rho)}\sqrt{2^{m-H_{\rm min}(Z^A|E)_\rho}}\sqrt{\delta}
\label{eq:original_LHL_dual_universal}
\end{equation}
for a sub-normalized $\rho_{Z^AE}$ \cite{FS08,6492260}.

On the other hand, if we instead apply the former half of Theorem \ref{thm:d_2_bound} to (\ref{eq:d_2_bound_almost_universal}), we obtain an LHL of the new type.
\begin{Lmm}[LHL for PA using an almost dual universal$_2$ hash function, in terms of $Q^{\rm PA}$]
\label{lmm:LHL_almost_universal}
For a $\delta$-almost dual universal$_2$ hash function $F$, and for a sub-normalized state $\rho_{Z^AE}$,
\begin{equation}
{\rm E}_F\, Q^{{\rm PA},F}(\rho_{Z^AE})\le 2^{m-H_{\rm min}(Z^A|E)_\rho}\delta.
\end{equation}
\end{Lmm}

Unlike the $\delta$-almost universal$_2$ function of the previous subsection (see the paragraph below Lemma \ref{lmm:LHL_QPA_almost_universal}), PA using this type of $F$ is secure even for $\delta$ exponentially larger than one.
Hence, this type of $F$ provides a much larger class of secure functions for PA.
Exploiting this property, in Ref. \cite{7399404} we proposed many useful examples of $F$, such as, efficiently computable hash functions requiring a small random seed.

\paragraph{Coding theorems}
Again by applying Theorem \ref{thm:equivalence_LHL_coding} to Lemma \ref{lmm:LHL_almost_universal}, we readily obtain the following lemma.
\begin{Lmm}[Coding theorems for EC and DC using an almost universal$_2$ function]
For a $\delta$-almost universal$_2$ function $G$, and for a sub-normalized state $\rho_{X^AB}$,
\begin{eqnarray}
{\rm E}_{G}\, Q^{{\rm EC},G}(\rho_{X^AB})&=&{\rm E}_{G}\, Q^{{\rm DC},G}(\rho_{X^AB})\\
&\le&  2^{H_{\rm max}(X^A|B)_\rho-(n-m)}\delta.\nonumber
\end{eqnarray}
\end{Lmm}
This lemma generalizes and improves the first inequality given in Theorem 1 of Ref. \cite{6157080}, which essentially says, in our notation,
\begin{equation}
{\rm E}_{G}\, Q^{{\rm DC},G}(\rho_{X^AB})\le 4 \sqrt{ 2^{H_{\rm max}(X^A|B)_\rho-(n-m)}}
\end{equation}
for normalized $\rho_{X^AB}$ and for $\delta=1$.

\section{Application 2: Equivalence of the two approaches to the security proof of quantum key distribution}

There are two major approaches for the security proof of quantum key distribution (QKD):
\begin{itemize}
\item {\bf Leftover hashing lemma (LHL)-based approach}, where one makes use of an LHL \cite{RennerPhD}. 
\item {\bf Phase error correction (PEC)-based approach} \cite{Mayers98,Lo2050,SP00,Koashi,H07}, where one transforms a given QKD protocol mathematically to a EC algorithm on the phase degree of freedom of its sifted key.
\end{itemize}
Previously we have proved that these two approaches are equivalent, in the sense that a proof of one approach can always be converted to the one of the other approach without affecting the resulting security bound \cite{8970489}.

Below, we will apply the results obtained in this paper to simplify and improve our previous proof in Ref. \cite{8970489}.
The proof below is improved in that it is valid for a larger classes of hash functions.
That is, the equivalence holds for the case where the random function $F$ for PA is {\it almost} universal$_2$ \cite{WEGMAN1981265} and {\it almost dual} universal$_2$ \cite{FS08,6492260}, while previously \cite{8970489} we treated only the case where $F$ is universal$_2$ \cite{CARTER1979143}.

Further, utilizing the knowledge gained in this new proof, we propose a method to simplify the PEC-based proof.
That is, we propose to evaluate the randomness of Alice's phase degrees of freedom by the smooth max-entropy, rather than by the phase error rate.
This method has an additional merit that every step of the proof becomes equivalent to that of  the corresponding LHL-based proof.
As a result, one is guaranteed to reach exactly the same security bound as in the LHL-based approach, without any extra factor.

\subsection{Quick review of the two approaches}
We begin by reviewing these approaches concentrating on typical cases, though not completely general.

\subsubsection{Typical QKD protocol}
\label{sec:common_properties}
First, we specify what we mean by a typical QKD protocol.

\paragraph{Participants and their degrees of freedom}
There are two legitimate users, Alice and Bob, who want to share a secure key, and an eavesdropper Eve.

We let space $A$ be Alice's $n$-qubit space for storing her sifted key, 
%This space need not necessarily be that of a real experimental device, but may be an ancillary space.
and $B$ be all other degrees of freedom of the legitimate users.
Note here that space $B$ generally includes Alice's degrees of freedom, besides Bob's.

We also let space $E$ be all degrees of freedom of Eve.

\paragraph{Protocol}
\label{sec:protocol_typical QKD}
For the sake of simplicity we limit ourselves with entanglement-based protocol.
We note that we do not lose generality here, since a prepare-and-measure protocols can always be transformed to an entanglement-based protocol by introducing ancillary spaces appropriately.

Initially, Alice, Bob, and Eve are in a state $\mu_{ABE}$.
Alice and Bob then perform the following protocol.
\begin{itemize}
\item[1.] {\bf Sample measurement} Alice and Bob measure their reduced state $\mu_{AB}$, and determine whether they abort the protocol or not.
\item[2.] {\bf Corrected key generation} Alice and Bob each performs a measurement in space $AB$ independently.
They then perform information reconciliation together, and generate their corrected keys $z\in\{0,1\}^n$ for Alice, and $z'\in\{0,1\}^n$ for Bob respectively.

We assume that Alice stores her corrected key $z$ in space $A$ in the $z$ basis, or in $Z^A$ basis.
We denote by $\rho_{Z^AE}$ the sub-normalized state after this step of Alice's $z$ and Eve's degree of freedom $E$.
\item[3.] {\bf Secret key generation} 
Alice and Bob each applies PA using a random function $F$ to their sifted keys $z$,  $z'$ respectively, and generates their secret keys $k,k'$.
\end{itemize}

\subsubsection{Security criterion of QKD}
\label{sec:security_QKD}
In order to show the security of a QKD protocol as a whole, it suffices to bound the sum of the failure probability of IR, $\varepsilon_{\rm IR}=\Pr[Z\ne Z']$ ($\ge \Pr[K\ne K']$), and the security index of Alice' secret key
${\rm E}_F\, d_1(\rho_{KE}^F)$ (see e.g. Remark 6.1.3 of Ref. \cite{RennerPhD}).

The probability $\varepsilon_{\rm IR}$ can easily be bounded using the theory of classical error correcting codes.

As a result, the security analysis of QKD is reduced to that of Alice's secret key, i.e., to bounding ${\rm E}_F\, d_1(\rho_{KE}^F)$.

\subsubsection{Leftover hashing lemma (LHL)-based approach}
\label{sec:LHL_based_approach}
For a QKD protocol of the above type, a typical LHL-based proof proceeds as follows.

\paragraph{Assumption}
One assumes that the sample measurement is designed so that the resulting sub-normalized state $\rho_{Z^AE}$ satisfies
\begin{equation}
H_{\rm min}^{\rm th}\le H^\varepsilon_{\rm min}(Z^A|E)_{\rho}
\label{eq:H_min_condition}
\end{equation}
with $H_{\rm min}^{\rm th}$ being a predetermined constant.

One also assumes that an LHL for the random function $F$ has already been proved.
In practice, it suffices to assume that Alice and Bob use one of the examples of $F$ that we studied in Section \ref{sec:application1_LHL_coding_theorem}, i.e.,  universal$_2$, $\delta$-almost universal$_2$ (with $\delta$ sufficiently small), and $\delta$-almost dual universal$_2$ functions.

\paragraph{Security proof}
\label{par:security_proof_LHL_based}
As mentioned in Section \ref{sec:security_QKD}, in order to prove the security of the QKD protocol, one needs to bound the security index of Alice's secret key, ${\rm E}_F\, d_1(\rho_{KE}^F)$.
For this purpose, one uses an LHL.

For example, if $F$ is universal$_2$, ${\rm E}_F\, d_1(\rho_{KE}^F)$ can be bound by using an LHL (\ref{eq:original_LHL}).
By substituting (\ref{eq:H_min_condition}) to (\ref{eq:original_LHL}), one obtains a security bound
\begin{equation}
{\rm E}_F\, d_1(\rho_{KE}^F)\le 2\varepsilon+\sqrt{2^{m-H_{\rm min}^{\rm th}}}.
\label{eq:LHL_based_bound}
\end{equation}

Also for the case where $F$ is $\delta$-almost universal$_2$, one can obtain a similar security bounds using an LHL (\ref{eq:LHL_conventional_almost_universal}).
For the case of $\delta$-almost universal$_2$, one can use (\ref{eq:original_LHL_dual_universal}).

\subsubsection{Phase error correction (PEC)-based approach}
\label{sec:PEC_based_approach}

On the other hand, a typical PEC-based proof proceeds as follows.
\paragraph{Assumption}
\label{par:assumption_PEC}
This approach starts by deriving a {\it virtual state} $\rho_{X^AB}$.
That is, one applies to state $\rho_{Z^AE}$ (defined in step 2 of Section \ref{sec:protocol_typical QKD}) the procedure given in Section \ref{sec:equivalence_three_algorithms} and Fig. \ref{fig:construction_standard_state}, and generates $\rho_{X^AB}$ along with a tripartite state $\rho_{ABE}$ in the standard form.

% revised 2021/10/07
Then one supposes that Alice and Bob together perform EC, with $\rho_{X^AB}$ regarded as the quantum channel on input Alice's codeword $c(t)=0$ and message $t=0$\footnote{
This corresponds to the situation where Alice always sends out $c(t)=0$, and Bob's goal is always to recover $y=0$.
%This setting originates from the early literature on the PEC approach \cite{Lo2050}, where the security was translated into entanglement distillation using quantum error correction.
If this setting seems unnatural, one may rewrite the situation such that the quantum channel is binary symmetric (see Section \ref{sec:EC_side_info})  by ``twirling'' or by randomizing the EC with bit flips in the $X$ basis (randomization by the Pauli $Z$ operators).
Such randomization is always possible since $\rho_{Z^AE}$ is invariant under the Pauli $Z$ operators.}.
% revised 2021/10/07
One also assumes that the sample measurement step and the random function $F$ are designed, so that the EC using the dual random function $F^\perp$ fails with a probability,
\begin{equation}
{\rm E}_F\,Q^{{\rm EC},F^\perp}(\rho_{X^AB})\le Q^{\rm EC,th},
\label{eq:Q_EC_condition}
\end{equation}
where $Q^{\rm EC,th}$ is a predetermined constant.

\paragraph{Security proof}
\label{sec:security_proof_PEC_based}
As mentioned in Section \ref{sec:security_QKD}, one needs to bound the security index of Alice's secret key, ${\rm E}_F\, d_1(\rho_{KE}^F)$.
For this purpose, one uses an inequality
\begin{equation}
d_1(\rho_{KE}^f)\le 2\sqrt2\sqrt{Q^{{\rm EC},f^{\perp}}(\rho_{X^AB})},
\label{eq:bound_on_d1_by_QEC}
\end{equation}
which has been known previously  (see e.g. Refs. \cite{HT12,8970489}).
By substituting condition (\ref{eq:Q_EC_condition}) to (\ref{eq:bound_on_d1_by_QEC}), one obtains a security bound
\begin{eqnarray}
{\rm E}_F\,d_1(\rho_{KE}^F)&\le& 2\sqrt2\sqrt{{\rm E}_F\,Q^{{\rm EC},F^\perp}(\rho_{X^AB})}\nonumber\\
&\le&2\sqrt2\sqrt{Q^{\rm EC,th}},
\label{eq:d1_bound_PEC}
\end{eqnarray}
where used Jensen's inequality is used in the first line.

\paragraph{Typical method for satisfying the assumption}
\label{eq:typical_way_fulfill}
In most literature of the PEC-based approach, one designs the sample measurement step and determines the upper bound $Q^{\rm EC,th}$ on the average failure probability of the virtual EC, as follows.
\begin{enumerate}
\item One chooses a POVM $N=\{N^e\,|\,\sum_eN^e=\II_B\}$ in space $B$, and calculates the classical distribution $p(x,e):={\rm Tr}\left\{\rho_{X^AB}\left(\ket{\widetilde{x}}\bra{\widetilde{x}}_A\otimes N^e\right)\right\}$.
%, instead of considering the original state $\rho_{X^AB}$.

\item One designs the sample measurement step such that the randomness of $x$ seen in $p(x,e)$ becomes sufficiently small, 
for example, such that the phase error rate $e_{\rm ph}:=\sum_{x\ne0,e}p(x,e)$ becomes sufficiently small.
%i.e., a classical distribution of the phase degrees of freedom $A$ and the measurement result $b$ in $B$.
\item Using the theory of classical EC, one obtains an upper bound on the average failure probability of the EC.
This bound serves as $Q^{\rm EC,th}$.
\end{enumerate}
The POVM $N$ here plays the same role as $M^s$ of Section \ref{sec:decoding_ec}; i.e., they both provide a hint to boost the performance of the decoder (cf. 3rd paragraph, Section \ref{sec:EC_side_info}).
The difference is that unlike $M^s$, $N$ does not depend on the syndrome $s=f^\perp(x)$.

\paragraph{Remarks}
There are two remarks regarding the above procedure.

First, in the context of the QKD, state $\rho_{X^AB}$ is often called a virtual state since it is not necessarily realized in the actual QKD protocol, unlike $\rho_{Z^AE}$.
It is rather a mathematical artifact introduced for simplifying security proofs.
For the same reason, the EC on the phase error correction considered above is often called the {\it virtual} EC.

Second, while the bound (\ref{eq:bound_on_d1_by_QEC}) has been known previously, it can alternatively be regarded as a consequence of the equivalence of EC and PA, which we have shown in Section \ref{sec:main_results}.
Indeed, if we combine (\ref{eq:d1_Q_bound}) and Theorem \ref{Thm:equivalence}, we obtain a slightly weaker inequality $d_1(\rho_{KE}^f)\le 4\sqrt{Q^{{\rm EC},f^{\perp}}(\rho_{X^AB})}$.

\subsection{How to convert one approach to the other}
\label{sec:how_to_convert_PEC_LHL}
By using our results of Section \ref{sec:main_results}, one can always convert proofs of these two approaches to each other, while keeping the bounds (\ref{eq:LHL_based_bound}) and (\ref{eq:d1_bound_PEC}) essentially the same.

\subsubsection{Conversion from the LHL-based to the PEC-based}
\label{sec:conversion_from_LHL_to_PEC}
When given a LHL-based proof, one can always convert it to an alternative proof of the PEC-based approach.

The basic idea is as follows.
When given $\rho_{Z^AE}$, one can always reconstruct the (sub-normalized) virtual state $\rho_{X^AB}$, by using the procedure of Fig. \ref{fig:construction_standard_state}.
If one applies EC to $\rho_{X^AB}$ thus obtained, the situation now becomes equivalent to the PEC-based approach.
In addition, the failure probability of the EC there equals the security of PA in the original LHL-based proof,
\begin{equation}
{\rm E}_F\,Q^{{\rm PA},F}(\rho_{Z^AE})={\rm E}_F\,Q^{{\rm EC},F^\perp}(\rho_{X^AB}).
\label{eq:equality_PA_EC_QKD}
\end{equation}
This is because, by definition, $\rho_{Z^AE}$ and $\rho_{X^AB}$ here are related via a standard form $\rho_{ABE}$, and Corollary \ref{crl:equivalence_random_case} can be applied.
If one further applies (\ref{eq:bound_on_d1_by_QEC}), one can recover essentially the same bound as the original LHL approach, though one is working in the PEC-based approach.

We will see this procedure in detail for the case where $F$ is universal$_2$.
(For the case where $F$ is almost universal$_2$ or almost dual universal$_2$, one also can perform a similar procedure using Lemma \ref{lmm:LHL_QPA_almost_universal} or \ref{lmm:LHL_almost_universal}).

By the definition of the smooth min-entropy, there exists a sub-normalized state $\bar{\rho}_{Z^AE}$ which is $\varepsilon$-close to $\rho_{Z^AE}$ and satisfies $H_{\rm min}(Z^A|E)_{\bar{\rho}}=H_{\rm min}^\varepsilon(Z^A|E)_\rho$.
Let $\bar{\rho}_{X^AB}$ be the virtual sub-normalized state corresponding to $\bar{\rho}_{Z^AE}$.
Then by applying Corollary \ref{crl:equivalence_random_case} and Lemma \ref{lmm:LHL1}, one obtains
\begin{eqnarray}
\lefteqn{{\rm E}_F\,Q^{{\rm EC},F^\perp}(\bar{\rho}_{X^AB})={\rm E}_F\,Q^{{\rm PA},F}(\bar{\rho}_{Z^AE})}\nonumber\\
&=& 2^{m-H_{\rm min}(Z^A|E)_{\bar{\rho}}}=2^{m-H_{\rm min}^\varepsilon(Z^A|E)_\rho}
\nonumber\\
&\le& 2^{m-H_{\rm min}^{\rm th}}.
\label{eq:bound_on_EFQPA_LHL_approach}
\end{eqnarray}

This means that the situation is now equivalent to that of the PEC-based approach where assumption (\ref{eq:Q_EC_condition}) holds for $\bar{\rho}_{X^AB}$ with $Q^{\rm EC,th}=2^{m-H_{\rm min}^{\rm th}}$.
Therefore by applying (\ref{eq:d1_bound_PEC}), one obtains a bound,
\begin{equation}
{\rm E}_F\,d_1(\bar{\rho}_{KB}^F)\le2\sqrt2 \sqrt{2^{\frac12(m-H_{\rm min}^{\rm th})}},
\end{equation}
and also the security bound for the actual state $\rho_{KB}^F$,
\begin{equation}
{\rm E}_F\,d_1(\rho_{KB}^F)\le2\varepsilon+2\sqrt2 \sqrt{2^{\frac12(m-H_{\rm min}^{\rm th})}},
\end{equation}
which is identical to (\ref{eq:LHL_based_bound}), except for the presence of the factor of $2\sqrt2$.

\subsubsection{Conversion from the PEC-based to the LHL-based approach}

Conversely, when given a security proof of the PEC-based approach, one can always convert it to a proof of the LHL-based approach.

To this end, one repeats the the reasoning of the first two paragraph of Section \ref{sec:conversion_from_LHL_to_PEC}, and reaches Eq. (\ref{eq:equality_PA_EC_QKD}).
Then by substituting condition (\ref{eq:Q_EC_condition}) to (\ref{eq:equality_PA_EC_QKD}), one obtains a LHL of the new type (see Section \ref{sec:LHL_new_type}),
\begin{equation}
{\rm E}_F\,Q^{{\rm PA},F}(\rho_{Z^AE})\le Q^{\rm EC,th}.
\label{eq:new_LHL_from_PEC}
\end{equation}

The situation is now equivalent to that of the LHL-based approach.
If one further applies (\ref{eq:d1_Q_bound}) to (\ref{eq:new_LHL_from_PEC}), one obtains
\begin{eqnarray}
{\rm E}_F\,d_1(\rho_{KE})&\le&4 \sqrt{{\rm Tr}\rho}\sqrt{{\rm E}_F\,Q^{{\rm PA},F}(\rho_{Z^AE})}\nonumber\\
&\le& 4\sqrt{Q^{\rm EC,th}},
\end{eqnarray}
which is the same as (\ref{eq:d1_bound_PEC}), except that it is looser by the factor of $\sqrt2$.

\subsection{Evaluating the phase randomness by the smooth max-entropy}
\label{sec:PEC_based_using_LHLs}

In Section \ref{eq:typical_way_fulfill}, we explained a typical method used in the PEC-based approach for designing the sample measurement step and for determining the bound $Q^{\rm EC,th}$ on the average failure probability of the virtual EC.

However, this method has a problem that there is no fixed methodology for finding the appropriate POVM $N$.
%The other is that the upper bound $Q^{\rm EC,th}$ cannot be tight in general (unless $\rho_{X^AB}$ is classical), since one does not analyze the quantum state $\rho_{X^AB}$, but instead analyze the classical probability distribution $p(x,b)$ of its measurement results.

We here propose a method to avoid this problem.
Namely, we point out that it is convenient to evaluate the randomness of the phase degrees of freedom $X^A$ by the conditional max-entropy $H_{\rm max}^\varepsilon(X^A|B)_{\rho}$, rather than by the phase error rate $e_{\rm ph}$ mentioned in Section \ref{eq:typical_way_fulfill}.
In this method one can exploit the coding theorems for EC (e.g. those derived in Section \ref{sec:explicit_example_LHL}) to determine $Q^{\rm EC,th}$, without being bothered by the choice of the POVN $N$.

This method has an additional merit that every step of the proof becomes equivalent to that of  the LHL-based approach of Section \ref{sec:LHL_based_approach}.
As a result, one is guaranteed to reach exactly the same security bound as in the LHL-based approach, without any extra factor.

\subsubsection{PEC-based approach using the smooth max-entropy}

\paragraph{Assumptions}
As in Section \ref{par:assumption_PEC}, we derive the (sub-normalized) virtual state $\rho_{X^AB}$ by the procedure given in Section \ref{sec:equivalence_three_algorithms} and Fig. \ref{fig:construction_standard_state}.
Then we assume the following two items.
\begin{enumerate}
\item The sample measurement step is designed such that the virtual (sub-normalized) state $\rho_{X^AB}$ ends up having the smooth min-entropy bounded as
\begin{equation}
H_{\rm max}^\varepsilon(X^A|B)_{\rho}\le n-H_{\rm min}^{\rm th}.
\label{eq:H_max_condition}
\end{equation}
\item The random hash function $F$ is universal$_2$.
\end{enumerate}

These two assumptions together guarantee that there exists a {\it approximate} virtual (sub-normalized) state $\bar{\rho}_{X^AB}$ ($\approx_\varepsilon\rho_{X^AB}$) which satisfies condition (\ref{eq:Q_EC_condition}) of the PEC-based approach, with the parameter
\begin{equation}
Q^{\rm EC,th}=2^{m-H_{\rm min}^{\rm th}}.
\label{eq:Q_EC_th_H_max_th}
\end{equation}

To see this, note that by the definition of the smooth max-entropy, there exists a sub-normalized state $\bar{\rho}_{X^AB}$  ($\approx_\varepsilon\rho_{X^AB}$) satisfying $H_{\rm max}(X^A|B)_{\bar{\rho}}=H_{\rm max}^\varepsilon(X^A|B)_\rho$.
Then by applying the coding theorem for EC using $F^\perp$ (Lemma \ref{lmm:coding_theorems_dual_univ}) to this $\bar{\rho}_{X^AB}$,
\begin{eqnarray}
{\rm E}_{F}\, Q^{{\rm EC},F^\perp}(\bar{\rho}_{X^AB})&\le& 2^{H_{\rm max}(X^A|B)_{\bar{\rho}}-(n-m)}\nonumber\\
&=& 2^{H_{\rm max}^\varepsilon(X^A|B)_\rho-(n-m)}\nonumber\\
&\le&2^{m-H_{\rm min}^{\rm th}}.
\label{eq:EF_QEC_bounded_by_Hmax}
\end{eqnarray}

\paragraph{Security proof}
Hence we can apply the PEC-based proof of Section \ref{sec:PEC_based_approach} to the approximate virtual (sub-normalized) state $\bar{\rho}_{X^AB}$.
By substituting (\ref{eq:EF_QEC_bounded_by_Hmax}) to (\ref{eq:bound_on_d1_by_QEC}), we obtain
\begin{equation}
{\rm E}_F\,d_1(\bar{\rho}_{KE}^F)\le2\sqrt2\sqrt{2^{H_{\rm max}^\varepsilon(X^A|B)_\rho-(n-m)}}.
\label{eq:E_F_d1_rho_bar}
\end{equation}
Since $\bar{\rho}_{KE}^f\approx_\varepsilon\rho_{KE}^f$, we then have an LHL which is expressed in terms of Alice's phase degree of freedom $X^A$, 
\begin{equation}
{\rm E}_F\,d_1(\rho_{KE}^F)\le2\varepsilon+2\sqrt2\sqrt{2^{H_{\rm max}^\varepsilon(X^A|B)_\rho-(n-m)}}.
\label{eq:result_new_PEC_approach1}
\end{equation}
By substituting condition (\ref{eq:H_max_condition}) to (\ref{eq:result_new_PEC_approach1}), we obtain
\begin{equation}
{\rm E}_F\,d_1(\rho_{KE}^F)\le 2\varepsilon+2\sqrt2 \sqrt{2^{m-H_{\rm min}^{\rm th}}}.
\label{eq:result_new_PEC_approach2}
\end{equation}

\subsubsection{Equivalence with the LHL-based approach}
The final result (\ref{eq:result_new_PEC_approach2}) of the PEC-based proof above is the same as (\ref{eq:LHL_based_bound}) of the LHL-based proof, except for the presence of the factor of $2\sqrt2$.
This is because we are in fact using the same assumption and the same inequality as in the LHL-based approach of Section \ref{sec:security_QKD}:
\begin{itemize}
\item Condition (\ref{eq:H_max_condition}) is equivalent to condition (\ref{eq:H_min_condition}) of the LHL-based approach.

This is because in the current situation, we have
\begin{equation}
H_{\rm min}^\varepsilon(Z^A|E)_\rho+H_{\rm max}^\varepsilon(X^A|B)_\rho= n
\label{eq:equality_uncertainty}
\end{equation}
due to Lemma \ref{lmm:uncertainty}.
\item The coding theorem (Lemma \ref{lmm:coding_theorems_dual_univ}) that we used to derive (\ref{eq:EF_QEC_bounded_by_Hmax}) is essentially equivalent to the LHL (\ref{eq:original_LHL}), which is used in the LHL-based approach.

To see this, we apply Theorem \ref{thm:equivalence_LHL_coding} (the equivalence of PA and EC) to Lemma \ref{lmm:coding_theorems_dual_univ}, and obtain an LHL of the new type,
\begin{eqnarray}
{\rm E}_{F}\, Q^{{\rm EC},F^\perp}(\rho_{X^AB})&=&{\rm E}_{F}\, Q^{{\rm PA},F}(\rho_{Z^AE})\nonumber\\
&\le& 2^{H_{\rm max}(X^A|B)_{\rho}-(n-m)}\nonumber\\
&=& 2^{m-H_{\rm min}(Z^A|E)_{\rho}}
\end{eqnarray}
If we then apply (\ref{eq:d1_Q_bound}) to these inequalities, we obtain an LHL of the conventional type, which differs from (\ref{eq:original_LHL}) by a factor of 4.
\end{itemize}

Thus, in fact our proof method here is essentially the same as the LHL-based approach that was specified in Section \ref{sec:security_QKD}.
The difference is that ours is re-formalized within the PEC-based approach.

\subsubsection{Advantages of our method}
From the standpoint of the PEC-based approach, we believe that this stronger version of the equivalence is an additional merit of our method.
As a result, the advantages of our method against the typical PEC-based approach can be summarized as follows.
\begin{enumerate}
\item It admits the use of the smoothing parameter $\varepsilon$.
\item The analysis is simple:
One need not specify POVM $N$, which was mentioned in Section \ref{eq:typical_way_fulfill}.
Once one finishes evaluating Alice's phase randomness in terms of $H^\varepsilon_{\rm max}(X^A|B)_{\rho}$, the security bound readily follows from the LHL (\ref{eq:result_new_PEC_approach1}).
\item The bound thus obtained is guaranteed to be the same as in the corresponding LHL-based proof.
\end{enumerate}

\subsubsection{Improved bound using the uncertainty relation}

Finally, we note that the LHL expressed in terms of Alice's phase degree of freedom, (\ref{eq:result_new_PEC_approach1}), is not by itself a new result (though we believe that our interpretation in the PEC-based approach is).
It is rather a direct consequence of the entropic uncertainty relation (the former half of Lemma \ref{lmm:uncertainty} of this paper), which was shown previously by Tomamichel and coauthors \cite{PhysRevLett.106.110506,2012NatCo...3..634T,TomamichelPhD}.
Moreover,  (\ref{eq:result_new_PEC_approach1}) can be improved by using their result.

If we substitute (\ref{eq:UCR}) to (\ref{eq:original_LHL}), we obtain an LHL
\begin{equation}
{\rm E}_F\, d_1(\rho_{KE}^F)\le2\varepsilon+\sqrt{{\rm Tr}(\rho)}\sqrt{2^{H_{\rm max}^\varepsilon(X^A|B)_\rho-(n-m)}},
\label{eq:LHL_in_PEC_based}
\end{equation}
which improves (\ref{eq:result_new_PEC_approach1}) by a factor of $2\sqrt2$.
If we then substitute condition (\ref{eq:H_max_condition}) to (\ref{eq:LHL_in_PEC_based}), we obtain a security bound 
\begin{equation}
{\rm E}_F\, d_1(\rho_{KE}^F)\le2\varepsilon+\sqrt{2^{m-H_{\rm min}^{\rm th}}},
\label{eq:security_bound_by_LHL_in_PEC}
\end{equation}
which again improves (\ref{eq:result_new_PEC_approach2}) by a factor of $2\sqrt2$, and equals (\ref{eq:LHL_based_bound}) of the LHL-based proof exactly.

\section{Summary and outlook}
%revised 2021/10/05 from here
We showed that quantum algorithms of privacy amplification (PA), error correction (EC), and data compression (DC) are equivalent, if we define the security of PA by using the purified distance and if we generalize EC and DC by adding quantum side information.
%revised 2021/10/05 to here

As an application of this equivalence, we took previously known security bounds of PA, and converted them into coding theorems for EC and DC that have not been obtained previously.
We applied these results to simplify and improve our previous result that the two prevalent approaches to the security proof of quantum key distribution (QKD) are equivalent.
We also propose a method to simplify the security proof of QKD by using the insight gained in this analysis.

An interesting future direction is to generalize the equivalence by using a dual pair of $K$-entropies \cite{PhysRevLett.108.210405} for defining the indices $Q$ for the three algorithms.
For example, if we use the von Neumann entropy to define the index $Q$ for any one of the three algorithms, and if we can still prove the equivalence, then all the three algorithms will share the same definition for their indices $Q$.
Therefore, the equivalence in a stronger sense will be established, and the three algorithms will become truly indistinguishable.

\appendices

\section{Proof of Theorem \ref{Thm:equivalence}}
\label{sec:proof_of_theorem}

In this section we will prove Theorem \ref{Thm:equivalence}.

Since all terms in (\ref{eq:Thm_equality}) are proportional to ${\rm Tr}(\rho_{ABE})$, it suffices to consider the case where $\rho_{ABE}$ is normalized.
In this case, Theorem \ref{Thm:equivalence} follows immediately by combining Eq. (\ref{eq:QPA_H_max}) and the following three lemmas.

\begin{Lmm}
\label{lmm:EC_DC_equivalence}
For a normalized state $\rho_{X^AB}$ and for a function $g$, we have
\begin{eqnarray}
\lefteqn{Q^{{\rm EC},g}(\rho_{X^AB})=Q^{{\rm DC},g}(\rho_{X^AB})}\nonumber\\
&=&1-2^{-H_{\rm min}(X^A|B,g(X^A))_\rho}.
\end{eqnarray}
\end{Lmm}

\begin{proof}
From the definitions of EC and DC, it is evident that their indices $Q^{{\rm EC},g}(\rho_{X^AB})$ and $Q^{{\rm DC},g}(\rho_{X^AB})$ can both be rewritten as
\begin{eqnarray}
\lefteqn{Q^{{\rm EC},g}(\rho_{X^AB})=Q^{{\rm DC},g}(\rho_{X^AB})}\nonumber\\
&=&\min_{\{M^s\}}\left(1-\sum_{x}{\rm tr}(\tilde{\rho}^{x} M^{g(x),x})\right).
\end{eqnarray}
Next note that the decoding of DC is equivalent to the situation where, given a state  
\begin{equation}
\tau=\sum_x\ket{\tilde{x}}\bra{\tilde{x}}_A\otimes\tilde{\rho}^{x}_B\otimes \ket{g(x)}\bra{g(x)}_D
\end{equation}
(with $D$ being a new ancillary space), one estimates $x$ by measuring spaces $B,D$.
According to Ref. \cite{5208530}, the success probability of this estimation equals $2^{H_{\rm min}(X^A|B,D)_{\tau}}$.
Then, by noting $H_{\rm min}(X^A|B,D)_{\tau}=H_{\rm min}(X^A|B,g(X^A))_{\rho}$, we obtain the lemma.
%The same result also follows from the dual problem of the definition of the min-entropy (see e.g. the dual problem on the right hand side of (4.5) of Tomamichel's PhD thesis \cite{TomamichelPhD}).
\end{proof}

\begin{Lmm}
\label{thm:quotient_inequality}
For a dual pair of functions $f,g$, and for a normalized state $\rho_{ABE}$, we have
\begin{equation}
Q^{{\rm PA},f}(\rho_{Z^AE})\le Q^{{\rm EC},g}(\rho_{X^AB}).
%= Q^{{\rm DC},g}(\rho_{X^AB}).
\label{eq:Hmin_ZKB}
\end{equation}
\end{Lmm}

\begin{proof}

We begin by introducing a new notation: We denote the $i$-th row of a matrix $f$ (appearing in Definition \ref{dfn:dual_functions}) by $f_i$, and components of $f_i$ by $f_{i1},\dots,f_{in}$.
In this notation, for example, the duality condition $fg^T=0$ can be written $f_i\cdot g_j=\sum_{k}f_{ik}g_{jk}=0$ for $\forall i,j$.

Next we note that the security index $Q^{{\rm PA},f}$ for PA can be rewritten as
\begin{eqnarray}
\lefteqn{Q^{{\rm PA},f}(\rho_{Z^AE})}\\
&=&1-\max_{\sigma\ge0,{\rm Tr}\sigma=1} F(\Pi^{{\rm PA},f}(\rho_{ABE})_{KE},2^{-m}\II_K\otimes \sigma_E)^2.\nonumber
\end{eqnarray}
by using the algorithm $\Pi^{{\rm PA},f}$ below, which is designed to affect spaces $A,E$ in the same way as the actual PA.
\begin{itemize}
\item {\bf Equivalent algorithm for PA ($\Pi^{{\rm PA},f}$):} Measure space $A$ using the operator $Z^{f_i}=Z^{f_{i1}}\otimes\cdots\otimes Z^{f_{in}}$ ($i=1,\dots,m$), and store the result $k\in\{0,1\}^m$ in space $K$.
\end{itemize}

Similarly, we also note that  the performance index $Q^{{\rm EC},g}(\rho_{X^AB})$ of EC can be rewritten as
\begin{eqnarray}
\lefteqn{Q^{{\rm EC},g}(\rho_{X^AB})}\nonumber\\
&=&1-\max_{\{M^s\}}F(\ket{\tilde{0}}\bra{\tilde{0}}_A,\Pi^{{\rm EC},g}(\rho_{ABE})_A)^2.
\label{eq:QEC_rewritten}
\end{eqnarray}
by using the algorithm $\Pi^{{\rm EC},g}$ below.
\begin{itemize}
\item {\bf Equivalent decoding algorithm for EC ($\Pi^{{\rm EC},g}$):}
\begin{enumerate}
\item {\bf Syndrome measurement:} Measure space $A$ using an operator $X^{g_i}=X^{g_{i1}}\otimes\cdots\otimes X^{g_{in}}$ ($i=1,\dots,n-m$), and record the result as the syndrome $s=(s_1,\dots,s_{n-m})\in\{0,1\}^{n-m}$.
\item {\bf Side information measurement:} Measure space $B$ using the POVM $M^s=\{M^{s,e}\}_{x'\in\{0,1\}^n}$, and obtain the estimated error pattern $e$.
\item {\bf Bit flip:} Apply the operator $Z^{e}$ in space $A$.
\end{enumerate}
\end{itemize}

The state $\ket{\tilde{0}}\bra{\tilde{0}}_A$ appearing in Eq. (\ref{eq:QEC_rewritten}) must take the form $\ket{\tilde{0}}\bra{\tilde{0}}_A\otimes\sigma_{E}$ for some $\sigma$, when it is extended to spaces $A,E$.
Thus Eq. (\ref{eq:QEC_rewritten}) can be rewritten further as
\begin{align}
&Q^{{\rm EC},g}(\rho_{X^AB})=\\
&1-\max_{\{M^s\},\sigma\ge0,\sigma,{\rm Tr}\sigma=1} F\left(\ket{\tilde{0}}\bra{\tilde{0}}_A\otimes\sigma_{E},\Pi^{{\rm EC},g}(\rho_{ABE})_{AE}\right)^2.\nonumber
\end{align}

Further,  if functions $f,g$ are dual, the operator $Z^{f_i}$ of $\Pi^{{\rm PA},f}$, commutes with $X^{g_i}$ and with $Z^{z'}$ of  $\Pi^{{\rm EC},g}$.
Thus algorithms $\Pi^{{\rm PA},f}$ and $\Pi^{{\rm EC},g}$ also commute with each other.
Therefore we have $\Pi^{{\rm PA},f}(\Pi^{{\rm EC},g}(\rho))_{KE}=\Pi^{{\rm EC},g}(\Pi^{{\rm PA},f}(\rho))_{KE}=\Pi^{{\rm PA},f}(\rho)_{KE}$, and
\begin{align}
&1-Q^{{\rm EC},g}(\rho_{X^AB})\\
&\le\max_{\{M^s\},\sigma\ge0,{\rm Tr}\sigma=1} \nonumber\\
&\quad F\left(\Pi^{{\rm PA},f}(\ket{\tilde{0}}\bra{\tilde{0}}_A\otimes\sigma_{E})_{KE},\Pi^{{\rm PA},f}(\Pi^{{\rm EC},g}(\rho)_{AE})_{KE}\right)^2\nonumber\\
&=\max_{\sigma\ge0,{\rm Tr}\sigma=1} F\left(\Pi^{{\rm PA},f}(\ket{\tilde{0}}\bra{\tilde{0}}_A\otimes\sigma_{E})_{KE},\Pi^{{\rm PA},f}(\rho)_{KE}\right)^2\nonumber\\
&=\max_{\sigma\ge0,{\rm Tr}\sigma=1}F\left(\II_{K}\otimes\sigma_{E}, \rho^{f}_{KE}\right)^2=1-Q^{{\rm PA},f}(\rho_{AE}).\nonumber
\end{align}

\end{proof}

\begin{Lmm}
For a dual pair of functions $f,g$, and for a normalized state $\rho_{ABE}$ in the standard form, we have
\begin{equation}
H_{\rm max}(f(Z^A)|E)_{\rho}+H_{\rm min}(X^A|B,g(X^A))_\rho=m.
\label{eq:UCR_fZ}
\end{equation}
\end{Lmm}

\begin{proof}
We continue to use the notation introduced in the proof of Lemma \ref{thm:quotient_inequality}.
In this notation, the random number $K=f(Z^A)$ in PA corresponds to measurement result by operators $(Z^{f_i})_A$, and the syndrome $S=g(X^A)$ in EC corresponds to that by $(X^{g_j})_A$.
Since these operators, $(Z^{f_i})_A$ and $(X^{g_j})_A$, commutes with each other for a dual pair $f,g$, we can decompose space $A$ into $A=A_1A_2$, such that random variable $K$ equals the result of the $z$ basis measurement in the $m$-qubit space $A_1$, and $S$ equals the result of the $x$ basis measurement in the $n-m$ qubit space $A_2$, i.e., $Z^{A_1}=K=f(Z^A)$ and $X^{A_2}=g(X^A)$.

In this case, $H_{\rm min}(X|B,X^{A_2})_\rho=H_{\rm min}(X^{A_1}|B,X^{A_2})_\rho$ holds, and the relation (\ref{eq:UCR_fZ}), which we need to prove, takes the form
\begin{equation}
H_{\rm max}(Z^{A_1}|E)_{\rho}+H_{\rm min}(X^{A_1}|B,X^{A_2})_\rho=m.
\label{eq:UCR_ZKB}
\end{equation}

(i) Suppose that $\rho_{AE}=\rho_{Z^AE}$.
In this case, if we measure $\ket{\rho}_{ABE}=\ket{\rho}_{A_1A_2BE}$ in the $X^{A_2}$ basis, we obtain
\begin{equation}
\rho_{A_1X^{A_2}BE}=2^{-(n-m)}\sum_{t}\ket{\tilde{s}}\bra{\tilde{s}}_{A_2}\otimes\ket{\tau^{s}}\bra{\tau^{s}}_{A_1BE}.
\label{eq:Th2_Ease1_1}
\end{equation}
Alternatively, if we measure $\rho_{A_1A_2E}=\rho_{Z^{A_1}Z^{A_2}E}$ in the $X^{A_2}$ basis, we obtain
\begin{equation}
\rho_{Z^{A_1}X^{A_2}E}=2^{-(n-m)}\II_{A_2}\otimes\rho_{Z^{A_1}E}.
\label{eq:Th2_Ease1_2}
\end{equation}
Then if we further trace out $B$ from (\ref{eq:Th2_Ease1_1}), it should equal (\ref{eq:Th2_Ease1_2}) by construction.
Thus all $\ket{\tau^{s}}$ are a purification of $\rho_{Z^{A_1}E}$.
Hence by applying Lemma \ref{lmm:uncertainty} to $\ket{\tau^{s}}$ and space $A_1$, we see that $H_{\rm min}(X^{A_1}|B)_{\tau^s}=m-H_{\rm max}(Z^{A_1}|E)_{\tau^s}=m-H_{\rm max}(Z^{A_1}|E)_{\rho}$.
Also by applying Proposition 4.6 of Ref. \cite{TomamichelPhD}, we obtain $2^{H_{\rm min}(X^{A_1}|X^{A_2},B)_\rho}=2^{-(n-m)}\sum_s2^{H_{\rm min}(X^{A_1}|B)_{\tau^s}}=2^{H_{\rm max}(Z^{A_1}|E)_{\rho}}$, which proves (\ref{eq:UCR_ZKB}).

(ii) Suppose that $\rho_{AB}=\rho_{X^AB}$.
Then we have $\rho_{A_1A_2B}=\rho_{A_1X^{A_2}B}=\rho_{X^{A_1}X^{A_2}B}$, with $\rho_{ABE}$ being the purification of all these states.
By applying Lemma \ref{lmm:uncertainty} to $\rho_{ABE}$ and space $A_1$, we obtain  (\ref{eq:UCR_ZKB}).
\end{proof}

\section{Relation between security criteria of privacy amplification}
\label{sec:security_Eriteria}

In section \ref{sec:PA}, we introduced quantities $Q^{{\rm PA},f}(\rho_{Z^AE})$ and $d_1(\rho_{KE}^f)$ as the security index for random number $K=f(Z^A)$ in PA.
Besides these two quantities, some literature also use  security index for $K$,
\begin{equation}
d_1'(\rho_{KE}^f):=\min_{\sigma\ge0,{\rm Tr}(\sigma)=1}\left\|\rho_{KE}^f-2^{-m}\II_K\otimes\sigma_E\right\|_1.
\end{equation}

It is straightforward to show that this quantity can be bounded by $Q^{{\rm PA},f}(\rho_{Z^AE})$ as
\begin{equation}
1-\sqrt{1-Q^{{\rm PA},f}(\rho_{Z^AE})}\le\frac12 d'_1(\rho_{KE}^f)\le \sqrt{Q^{{\rm PA},f}(\rho_{Z^AE})},
\label{eq:bound_d1'}
\end{equation}
if $\rho$ is normalized  (see from Eq. (9.110) of Ref. \cite{Nielsen-Chuang}).

It can also be shown that $d_1'(\rho_{KE}^f)$ bounds the conventional security index $d_1(\rho_{KE}^f)$ as
\begin{equation}
d_1'(\rho_{KE}^f)\le d_1(\rho_{KE}^f)\le2d_1'(\rho_{KE}^f),
\label{eq:bound_d1}
\end{equation}
if $\rho$ is normalized.
The first inequality is immediate from the definitions of  $d_1(\rho_{KE}^f)$ and  $d_1'(\rho_{KE}^f)$.
The second inequality can be shown as
\begin{eqnarray}
\lefteqn{\left\|\rho_{KE}^f-2^{-m}\II_K\otimes\rho_E\right\|}\nonumber\\
&\le&
\left\|\rho_{KE}^f-2^{-m}\II_K\otimes\sigma_E\right\|\nonumber\\
&&+\left\|2^{-m}\II_K\otimes\sigma_E-2^{-m}\II_K\otimes\rho_E\right\|
\nonumber\\
&=&\left\|\rho_{KE}^f-2^{-m}\II_K\otimes\sigma_E\right\|+
\left\|\rho_E-\sigma_E\right\|\nonumber\\
&\le&2\left\|\rho_{KE}^f-2^{-m}\II_K\otimes\sigma_E\right\|,
\end{eqnarray}
where the last inequality holds from the monotonicity of the trace distance.

Inequality (\ref{eq:d1_Q_bound}) of section \ref{sec:PA} follows from Eqs. (\ref{eq:bound_d1'}) and (\ref{eq:bound_d1}).

\section{A note on the correspondence of states}
\label{sec:note_on_correspondence}

In Sec. \ref{par:correspondence_states}, we introduced the procedure $T$ in order to define the correspondence between the input state $\rho_{Z^AE}$ of PA, and $\rho_{X^AB}$ of EC and DC.
However, contrary to our naive expectation, repeating this procedure $T$ twice does not in fact yield the original state:
A straightforward calculation gives
\begin{eqnarray}
T(T(\rho_{Z^AE}))&=&{\rm XOR}^z_{E'\to A}(\rho_{Z^AE}\otimes 2^{-n}\II_{E'}),
\label{eq:TTrhoZAE}\\
T(T(\rho_{X^AB}))&=&{\rm XOR}^x_{B'\to A}(\rho_{X^AB}\otimes 2^{-n}\II_{B'}),
\end{eqnarray}
where the Hilbert spaces $E'$ and $B'$ are of the same size as $A$, and ${\rm XOR}^b_{C\to A}$ denotes XORing variable $C$ on variable $A$ in basis $b\in\{z,x\}$.
Note that $T(T(\rho_{Z^AE}))\ne\rho_{Z^AE}$ and $T(T(\rho_{X^AB}))\ne\rho_{X^AB}$.

Here we demonstrate that this fact does not compromise the correspondence of the two types of states.
The basic observation is that states $T(T(\rho_{Z^AE}))$ and $\rho_{Z^AE}$ are equivalent as long as one is concerned with the security of PA;
and $T(T(\rho_{X^AB}))$ and $\rho_{X^AB}$ are equivalent as long as the performance of DC or EC is concerned.

To see this for the case of $\rho_{Z^AE}$, note that the state ${\rm XOR}^z_{E'\to A}(\rho_{Z^AE}\otimes 2^{-n}\II_{E'})$ appearing in (\ref{eq:TTrhoZAE}) describes a classical ensemble of $\rho_{Z^AE}$ where Alice's classical variable $Z^A$ is shifted (XORed) by a publicly known random variable $E'$.
Hence performing PA on this state is equivalent to performing it on $\rho_{Z^AE}$ in parallel, where Alice's random bits $k$ is shifted by a public constant value $f(e')$.
Since the security of $k$ cannot be affected by such shift, this situation is clearly equivalent to PA on $\rho_{Z^AE}$.

We note that essentially the same argument also holds for the case of $\rho_{X^AB}$ in DC and EC.

\bibliographystyle{IEEEtran}
\bibliography{CSI_PA_equivalence}

\begin{thebibliography}{10}
\providecommand{\url}[1]{#1}
\csname url@samestyle\endcsname
\providecommand{\newblock}{\relax}
\providecommand{\bibinfo}[2]{#2}
\providecommand{\BIBentrySTDinterwordspacing}{\spaceskip=0pt\relax}
\providecommand{\BIBentryALTinterwordstretchfactor}{4}
\providecommand{\BIBentryALTinterwordspacing}{\spaceskip=\fontdimen2\font plus
\BIBentryALTinterwordstretchfactor\fontdimen3\font minus
  \fontdimen4\font\relax}
\providecommand{\BIBforeignlanguage}[2]{{%
\expandafter\ifx\csname l@#1\endcsname\relax
\typeout{** WARNING: IEEEtran.bst: No hyphenation pattern has been}%
\typeout{** loaded for the language `#1'. Using the pattern for}%
\typeout{** the default language instead.}%
\else
\language=\csname l@#1\endcsname
\fi
#2}}
\providecommand{\BIBdecl}{\relax}
\BIBdecl

\bibitem{RennerPhD}
R.~Renner, ``Security of quantum key distribution,'' Ph.D. dissertation, Diss.
  ETH No. 16242, 2005.

\bibitem{VanAssche}
G.~V. Assche, \emph{Quantum Cryptography and Secret-Key Distillation}.\hskip
  1em plus 0.5em minus 0.4em\relax Cambridge University Press, 2006.

\bibitem{tuyls2007security}
\BIBentryALTinterwordspacing
P.~Tuyls, B.~{\v{S}}koric, and T.~Kevenaar, \emph{Security with Noisy Data: On
  Private Biometrics, Secure Key Storage and Anti-Counterfeiting}, ser.
  Security with noisy data.\hskip 1em plus 0.5em minus 0.4em\relax Springer
  London, 2007. [Online]. Available:
  \url{https://books.google.co.jp/books?id=psWUQ5cDsI0C}
\BIBentrySTDinterwordspacing

\bibitem{MYCQZ16}
\BIBentryALTinterwordspacing
X.~Ma, X.~Yuan, Z.~Cao, B.~Qi, and Z.~Zhang, ``Quantum random number
  generation,'' \emph{Npj Quantum Information}, vol.~2, p. 16021, Jun 2016.
  [Online]. Available: \url{https://doi.org/10.1038/npjqi.2016.21}
\BIBentrySTDinterwordspacing

\bibitem{10.5555/1146355}
T.~M. Cover and J.~A. Thomas, \emph{Elements of Information Theory (Wiley
  Series in Telecommunications and Signal Processing)}.\hskip 1em plus 0.5em
  minus 0.4em\relax USA: Wiley-Interscience, 2006.

\bibitem{5961850}
M.~{Tomamichel}, C.~{Schaffner}, A.~{Smith}, and R.~{Renner}, ``Leftover
  hashing against quantum side information,'' \emph{IEEE Transactions on
  Information Theory}, vol.~57, no.~8, pp. 5524--5535, 2011.

\bibitem{doi:10.1098/rspa.2010.0445}
\BIBentryALTinterwordspacing
J.~M. Renes, ``Duality of privacy amplification against quantum adversaries and
  data compression with quantum side information,'' \emph{Proceedings of the
  Royal Society A: Mathematical, Physical and Engineering Sciences}, vol. 467,
  no. 2130, pp. 1604--1623, 2011. [Online]. Available:
  \url{https://royalsocietypublishing.org/doi/abs/10.1098/rspa.2010.0445}
\BIBentrySTDinterwordspacing

\bibitem{8970489}
T.~{Tsurumaru}, ``Leftover hashing from quantum error correction: Unifying the
  two approaches to the security proof of quantum key distribution,''
  \emph{IEEE Transactions on Information Theory}, vol.~66, no.~6, pp.
  3465--3484, 2020.

\bibitem{5208530}
R.~{K\"{o}nig}, R.~{Renner}, and C.~{Schaffner}, ``The operational meaning of
  min- and max-entropy,'' \emph{IEEE Transactions on Information Theory},
  vol.~55, no.~9, pp. 4337--4347, 2009.

\bibitem{TomamichelPhD}
M.~Tomamichel, ``A framework for non-asymptotic quantum information theory,''
  Ph.D. dissertation, Diss. ETH No. 20213, 2012.

\bibitem{CARTER1979143}
\BIBentryALTinterwordspacing
J.~Carter and M.~N. Wegman, ``Universal classes of hash functions,''
  \emph{Journal of Computer and System Sciences}, vol.~18, no.~2, pp. 143 --
  154, 1979. [Online]. Available:
  \url{http://www.sciencedirect.com/science/article/pii/0022000079900448}
\BIBentrySTDinterwordspacing

\bibitem{WEGMAN1981265}
\BIBentryALTinterwordspacing
M.~N. Wegman and J.~Carter, ``New hash functions and their use in
  authentication and set equality,'' \emph{Journal of Computer and System
  Sciences}, vol.~22, no.~3, pp. 265 -- 279, 1981. [Online]. Available:
  \url{http://www.sciencedirect.com/science/article/pii/0022000081900337}
\BIBentrySTDinterwordspacing

\bibitem{FS08}
S.~Fehr and C.~Schaffner, ``Randomness extraction via delta-biased masking in
  the presence of a quantum attacker,'' \emph{Theory of Cryptography Fifth
  Theory of Cryptography Conference, TCC 2008 New York, USA, March 19-21,
  Lecture Notes in Computer Science}, vol. 4948, pp. 465--481, 2008.

\bibitem{6492260}
T.~{Tsurumaru} and M.~{Hayashi}, ``Dual universality of hash functions and its
  applications to quantum cryptography,'' \emph{IEEE Transactions on
  Information Theory}, vol.~59, no.~7, pp. 4700--4717, 2013.

\bibitem{Nielsen-Chuang}
M.~A. Nielsen and I.~L. Chuang, \emph{Quantum Computation and Quantum
  Information: 10th Anniversary Edition}, 10th~ed.\hskip 1em plus 0.5em minus
  0.4em\relax New York, NY, USA: Cambridge University Press, 2011.

\bibitem{Mayers98}
\BIBentryALTinterwordspacing
D.~Mayers, ``Unconditional security in quantum cryptography,'' \emph{J. ACM},
  vol.~48, no.~3, pp. 351--406, May 2001. [Online]. Available:
  \url{http://doi.acm.org/10.1145/382780.382781}
\BIBentrySTDinterwordspacing

\bibitem{Lo2050}
\BIBentryALTinterwordspacing
H.-K. Lo and H.~F. Chau, ``Unconditional security of quantum key distribution
  over arbitrarily long distances,'' \emph{Science}, vol. 283, no. 5410, pp.
  2050--2056, 1999. [Online]. Available:
  \url{https://science.sciencemag.org/content/283/5410/2050}
\BIBentrySTDinterwordspacing

\bibitem{SP00}
\BIBentryALTinterwordspacing
P.~W. Shor and J.~Preskill, ``Simple proof of security of the {BB}84 quantum
  key distribution protocol,'' \emph{Phys. Rev. Lett.}, vol.~85, pp. 441--444,
  Jul 2000. [Online]. Available:
  \url{https://link.aps.org/doi/10.1103/PhysRevLett.85.441}
\BIBentrySTDinterwordspacing

\bibitem{Koashi}
M.~Koashi, ``Complementarity, distillable secret key, and distillable
  entanglement,'' \emph{arXiv:0704.3661 [quant-ph]}, 2007.

\bibitem{H07}
\BIBentryALTinterwordspacing
M.~Hayashi, ``Upper bounds of eavesdropper's performances in finite-length code
  with the decoy method,'' \emph{Phys. Rev. A}, vol.~76, p. 012329, Jul 2007.
  [Online]. Available:
  \url{https://link.aps.org/doi/10.1103/PhysRevA.76.012329}
\BIBentrySTDinterwordspacing

\bibitem{PhysRevA.78.032335}
\BIBentryALTinterwordspacing
J.~M. Renes and J.-C. Boileau, ``Physical underpinnings of privacy,''
  \emph{Phys. Rev. A}, vol.~78, p. 032335, Sep 2008, (preprints
  arXiv:quant-ph/0702187v1 and arXiv:0803.3096v2 [quant-ph]). [Online].
  Available: \url{https://link.aps.org/doi/10.1103/PhysRevA.78.032335}
\BIBentrySTDinterwordspacing

\bibitem{6157080}
J.~M. {Renes} and R.~{Renner}, ``One-shot classical data compression with
  quantum side information and the distillation of common randomness or secret
  keys,'' \emph{IEEE Transactions on Information Theory}, vol.~58, no.~3, pp.
  1985--1991, 2012.

\bibitem{8047296}
J.~M. Renes, ``Duality of channels and codes,'' \emph{IEEE Transactions on
  Information Theory}, vol.~64, no.~1, pp. 577--592, 2018.

\bibitem{8434309}
------, ``On privacy amplification, lossy compression, and their duality to
  channel coding,'' \emph{IEEE Transactions on Information Theory}, vol.~64,
  no.~12, pp. 7792--7801, 2018.

\bibitem{10.1007/978-3-540-30576-7_21}
M.~Ben-Or, M.~Horodecki, D.~W. Leung, D.~Mayers, and J.~Oppenheim, ``The
  universal composable security of quantum key distribution,'' in \emph{Theory
  of Cryptography}, J.~Kilian, Ed.\hskip 1em plus 0.5em minus 0.4em\relax
  Berlin, Heidelberg: Springer Berlin Heidelberg, 2005, pp. 386--406.

\bibitem{2018arXiv180911143C}
H.-C. {Cheng}, E.~P. {Hanson}, N.~{Datta}, and M.-H. {Hsieh}, ``{Duality
  between source coding with quantum side information and c-q channel
  coding},'' \emph{arXiv e-prints}, p. arXiv:1809.11143, Sep. 2018.

\bibitem{PhysRevA.68.042301}
\BIBentryALTinterwordspacing
I.~Devetak and A.~Winter, ``Classical data compression with quantum side
  information,'' \emph{Phys. Rev. A}, vol.~68, p. 042301, Oct 2003. [Online].
  Available: \url{https://link.aps.org/doi/10.1103/PhysRevA.68.042301}
\BIBentrySTDinterwordspacing

\bibitem{9261419}
H.-C. Cheng, E.~P. Hanson, N.~Datta, and M.-H. Hsieh, ``Non-asymptotic
  classical data compression with quantum side information,'' \emph{IEEE
  Transactions on Information Theory}, vol.~67, no.~2, pp. 902--930, 2021.

\bibitem{7399404}
M.~{Hayashi} and T.~{Tsurumaru}, ``More efficient privacy amplification with
  less random seeds via dual universal hash function,'' \emph{IEEE Transactions
  on Information Theory}, vol.~62, no.~4, pp. 2213--2232, 2016.

\bibitem{PhysRevLett.106.110506}
\BIBentryALTinterwordspacing
M.~Tomamichel and R.~Renner, ``Uncertainty relation for smooth entropies,''
  \emph{Phys. Rev. Lett.}, vol. 106, p. 110506, Mar 2011. [Online]. Available:
  \url{https://link.aps.org/doi/10.1103/PhysRevLett.106.110506}
\BIBentrySTDinterwordspacing

\bibitem{PhysRevLett.108.210405}
\BIBentryALTinterwordspacing
P.~J. Coles, R.~Colbeck, L.~Yu, and M.~Zwolak, ``Uncertainty relations from
  simple entropic properties,'' \emph{Phys. Rev. Lett.}, vol. 108, p. 210405,
  May 2012. [Online]. Available:
  \url{https://link.aps.org/doi/10.1103/PhysRevLett.108.210405}
\BIBentrySTDinterwordspacing

\bibitem{TBH14}
\BIBentryALTinterwordspacing
M.~Tomamichel, M.~Berta, and M.~Hayashi, ``Relating different quantum
  generalizations of the conditional {R}\'{e}nyi entropy,'' \emph{Journal of
  Mathematical Physics}, vol.~55, no.~8, p. 082206, 2014. [Online]. Available:
  \url{https://doi.org/10.1063/1.4892761}
\BIBentrySTDinterwordspacing

\bibitem{HT12}
\BIBentryALTinterwordspacing
M.~Hayashi and T.~Tsurumaru, ``Concise and tight security analysis of the
  {B}ennett-{B}rassard 1984 protocol with finite key lengths,'' \emph{New
  Journal of Physics}, vol.~14, no.~9, p. 093014, 2012. [Online]. Available:
  \url{http://stacks.iop.org/1367-2630/14/i=9/a=093014}
\BIBentrySTDinterwordspacing

\bibitem{2012NatCo...3..634T}
M.~{Tomamichel}, C.~C.~W. {Lim}, N.~{Gisin}, and R.~{Renner}, ``{Tight
  finite-key analysis for quantum cryptography},'' \emph{Nature
  Communications}, vol.~3, p. 634, Jan. 2012.

\end{thebibliography}

\begin{IEEEbiographynophoto}{Toyohiro Tsurumaru} was born in Japan in 1973.
He received the B.S. degree from the Faculty of Science, University of Tokyo, Japan in 1996,
and the M.S. and Ph.D. degrees in physics from the Graduate School of Science, University of Tokyo, Japan in 1998 and 2001, respectively.
Then he joined Mitsubishi Electric Corporation in 2001.
His research interests include theoretical aspects of quantum cryptography, as well as modern cryptography.
\end{IEEEbiographynophoto}

\end{document}